\documentclass[a4paper,onecolumn,unpublished]{quantumarticle}
\pdfoutput=1

\usepackage[utf8]{inputenc}
\usepackage[english]{babel}
\usepackage[T1]{fontenc}
\usepackage{amsmath}
\usepackage{mathrsfs}       
\usepackage{mathtools}
\usepackage{amsfonts}   
\usepackage[titletoc,title]{appendix}
\usepackage{bbm}            
\usepackage{bm}             
\usepackage{comment}
\usepackage{amsthm}         
\usepackage{makecell}
\usepackage{algorithm} 
\usepackage{algorithmicx} 
\usepackage{algpseudocode} 
\usepackage{caption}
\makeatletter
\renewcommand{\ALG@name}{Protocol}
\makeatother
\usepackage[dvipsnames]{xcolor}
\definecolor{beamer@blendedblue}{rgb}{0.2,0.2,0.7}
\definecolor{googleblue}{HTML}{4285F4}
\definecolor{googlered}{HTML}{DB4437}
\definecolor{googleyellow}{HTML}{F4B400}
\definecolor{googlegreen}{HTML}{0F9D58}
\definecolor{klevinblue}{HTML}{002FA7}
\definecolor{tiffanyblue}{HTML}{0ABAB5}

\usepackage{booktabs} 

\usepackage[colorlinks=true,%
      linkcolor=beamer@blendedblue,%
      bookmarks=true,%
      breaklinks=true,%
      filecolor=beamer@blendedblue,%
      anchorcolor=yellow,%
      citecolor=beamer@blendedblue,%
      urlcolor=beamer@blendedblue,%
      pdfauthor={{TBD}},%
      pdfsubject={{TBD}},%
      CJKbookmarks=true]{hyperref}

\usepackage{float}    
\usepackage{verbatim}
\usepackage{tikz}
\usetikzlibrary{positioning, arrows.meta}

\usepackage{quantikz}

\newtheorem{definition}{Definition}
\newtheorem{proposition}[definition]{Proposition}

\newtheorem{lemma}[definition]{Lemma}
\newtheorem{theorem}[definition]{Theorem}

\mathchardef\ordinarycolon\mathcode`\:
\mathcode`\:=\string"8000
\def\vcentcolon{\mathrel{\mathop\ordinarycolon}}
\begingroup \catcode`\:=\active
  \lowercase{\endgroup
  \let :\vcentcolon
  }

\DeclareFontFamily{U}{mathx}{\hyphenchar\font45}
\DeclareFontShape{U}{mathx}{m}{n}{<-> mathx10}{}
\DeclareSymbolFont{mathx}{U}{mathx}{m}{n}
\DeclareMathAccent{\widebar}{0}{mathx}{"73}

\newcommand{\wt}[1]{\widetilde{#1}}
\newcommand{\wh}[1]{\widehat{#1}}

\newcommand{\ketbra}[2]{\vert{#1}\rangle\!\langle{#2}\vert}
\newcommand{\linear}[1]{\mathscr{L}(#1)}

\newcommand{\opn}[1]{\operatorname{#1}}

\DeclareMathOperator{\tr}{Tr}  

\newcommand{\ox}{\otimes}

\newcommand{\norm}[2]{\ensuremath{\left\lVert#1\right\rVert_{#2}}}%

\makeatletter
\newsavebox{\@brx}
\newcommand{\llangle}[1][]{\savebox{\@brx}{\(\m@th{#1\langle}\)}%
  \mathopen{\copy\@brx\kern-0.5\wd\@brx\usebox{\@brx}}}
\newcommand{\rrangle}[1][]{\savebox{\@brx}{\(\m@th{#1\rangle}\)}%
  \mathclose{\copy\@brx\kern-0.5\wd\@brx\usebox{\@brx}}}
\makeatother

\newcommand*{\cC}{\mathcal{C}}
\newcommand*{\cD}{\mathcal{D}}

\newcommand*{\cF}{\mathcal{F}}

\newcommand*{\cH}{\mathcal{H}}

\newcommand*{\cJ}{\mathcal{J}}

\newcommand*{\cL}{\mathcal{L}}
\newcommand*{\cM}{\mathcal{M}}
\newcommand*{\cN}{\mathcal{N}}
\newcommand*{\cO}{\mathcal{O}}

\newcommand*{\cS}{\mathcal{S}}

\newcommand*{\cU}{\mathcal{U}}

\newcommand*{\cW}{\mathcal{W}}

\newcommand{\bE}{\mathbb{E}}
\newcommand{\bZ}{\mathbb{Z}}

\usepackage{enumitem,amssymb,pifont}
\newlist{todolist}{itemize}{2}
\setlist[todolist]{label=$\square$}
\makeatletter
\newcommand{\appendixtitle}[1]{\gdef\@title{#1}}
\newcommand{\appendixauthor}[1]{\gdef\@author{#1}}
\newcommand{\appendixaffiliation}[1]{\gdef\@affiliation{#1}}
\newcommand{\appendixdate}[1]{\gdef\@date{#1}}
\makeatother

\makeatletter
\newcommand{\appendixmaketitle}{%
  \begin{center}%
    {\Large \@title \par}%
    \vspace{10pt}
    {\normalsize
      \lineskip .5em%
      \begin{tabular}[t]{c}%
        \@author
      \end{tabular}\par}%
    \vspace{5pt}
    {\itshape \@affiliation \par}%
    {\normalsize \@date}%
  \end{center}%
  \par
}
\makeatother

\definecolor{wildstrawberry}{rgb}{1.0, 0.26, 0.64}
\definecolor{googleblue}{HTML}{4285F4}
\definecolor{googlered}{HTML}{DB4437}
\definecolor{googleyellow}{HTML}{F4B400}
\definecolor{googlegreen}{HTML}{0F9D58}
\definecolor{klevinblue}{HTML}{002FA7}
\definecolor{tiffanyblue}{HTML}{0ABAB5}

\newcommand{\STAB}{{{\operatorname{STAB}}}}
\newcommand{\bu}{{{\textbf{u}}}}

\begin{document}

\title{{Quantum Fidelity Estimation in the Resource Theory of Nonstabilizerness}}

\author{Zhiping Liu}

\affiliation{Thrust of Artificial Intelligence, Information Hub,
The Hong Kong University of Science and Technology (Guangzhou), Guangdong 511453, China}
\affiliation{National Laboratory of Solid State Microstructures, School of Physics and Collaborative Innovation Center of Advanced Microstructures, Nanjing University, Nanjing 210093, China}
\author{Kun Wang}
\email{nju.wangkun@gmail.com}
\affiliation{Thrust of Artificial Intelligence, Information Hub,
The Hong Kong University of Science and Technology (Guangzhou), Guangdong 511453, China}

\author{Xin Wang}
\email{felixxinwang@hkust-gz.edu.cn}
\affiliation{Thrust of Artificial Intelligence, Information Hub,
The Hong Kong University of Science and Technology (Guangzhou), Guangdong 511453, China}

\fontfamily{lmr}\selectfont

\begin{abstract}
Quantum fidelity estimation is essential for benchmarking quantum states and processes on noisy quantum devices. While stabilizer operations form the foundation of fault-tolerant quantum computing, non-stabilizer resources further enable universal quantum computation through state injection. In this work, we propose several efficient fidelity estimation protocols for both quantum states and channels within the resource theory of nonstabilizerness, focusing on qudit systems with odd prime dimensions. Our protocols require measuring only a constant number of phase-space point operator expectation values, with operators selected randomly according to an importance weighting scheme tailored to the target state. Notably, we demonstrate that mathematically defined nonstabilizerness measures—such as Wigner rank and mana—quantify the sample complexity of the proposed protocols, thereby endowing them with a clear operational interpretation in the fidelity estimation task. This connection reveals a fundamental trade-off: while fidelity estimation for general quantum states and channels requires resources that scale exponentially with their nonstabilizerness, the task remains tractable for states and channels that admit efficient classical simulation.
\end{abstract}

\maketitle
\tableofcontents

\setlength{\parskip}{\baselineskip}


\section{Introduction}\label{sec:introduction}           
In the realm of quantum information processing, delicate quantum systems offer advantages over classical counterparts~\cite{de2021materials, arute2019quantum, zhong2020quantum, wu2021strong}.
However, noise from imperfect quantum hardware and operations challenges the accurate preparation, manipulation, and verification of quantum states~\cite{clerk2010introduction, wang2023mitigating,sun2024sudden}. To ensure reliable performance, it is essential to characterize experimentally prepared states and estimate their fidelity for the target state. Conventional quantum state tomography~\cite{Renes_2004,Lvovsky_2009, cramer2010efficient} addressed this challenge but suffers from the expensive resource growing exponentially with the size of the system, rendering it infeasible for large-scale quantum systems even with advancements such as compressed sensing techniques~\cite {Gross_2010, kalev2015quantum, riofrio2017experimental}. To overcome this limitation, direct fidelity estimation (DFE)~\cite{Flammia-2011-Phys.Rev.Lett.,Silva-2011-Phys.Rev.Lett.} and other related estimation methods~\cite{Elben-2020-Phys.Rev.Lett.,Zhu-2022-Nat.Commun.,Anshu-2022-STOC2022,Zheng-2024-NpjQuantumInf.} have emerged as promising alternatives. In particular, DFE entails the estimation of the fidelity between two quantum states without the need for intermediate state tomography, offering a more efficient and practical means of assessing the quality of quantum states and operations. Moreover, recent advances~\cite {Zhang_2021,qin2024experimental} have incorporated machine learning techniques to further enhance the efficiency of DFE. 

The efficiency of DFE protocols is typically measured by their sample complexity, i.e., the number of consumed quantum states or the number of calls to the unknown quantum channel required for accurate estimation. However, one might ask whether there exists another complementary perspective for assessing the feasibility of a DFE protocol. Since DFE protocols aim to extract inherently quantum information from the target states, it is natural to expect that their implementation efficiency is closely related to the underlying quantum resource properties of these states. Quantum resource theories~\cite{chitambar2019quantum} provide a powerful framework to rigorously characterize and quantify intrinsic quantum resources for various information-processing tasks~\cite{takagi2019general,takagi2020application,liu2024quantum, zhu2024limitations}, offering a promising lens to explore this potential connection.

In particular, within the context of fault-tolerant quantum computation (FTQC)~\cite{shor1996fault,gottesman1998theory, campbell2017roads}, the nonstabilizerness has emerged as a key resource~\cite{Campbell_2017}. While stabilizer operations offer a fault-tolerant approach to quantum computation, the Gottesman-Knill Theorem reveals that stabilizer circuits, composed of Clifford gates, can be efficiently simulated classically~\cite{gottesman2009introduction}. To achieve universal quantum computation and unlock its potential quantum advantage, non-stabilizer operations containing nonstabilizerness resources are required. Given this fundamental role, it is natural to ask: \emph{Does the nonstabilizerness of quantum states and operations also govern the feasibility of direct fidelity estimation?} Addressing this question has the potential to significantly enhance both practical quantum certification methods and theoretical insights into quantum resource theories. A recent work~\cite{Leone_2023} has made valuable explorations on this avenue in the multi-qubit systems based on a popular magic measure, namely stabilizer R\'enyi entropy~\cite{Leone2022,haug2023stabilizer, zhu2024amortized}.

Meanwhile, high-dimensional quantum systems, namely qudits, are emerging as powerful platforms for quantum computation and communication, with the potential to surpass the capabilities of traditional qubit-based technologies~\cite{Wang_2020, Ringbauer_2022, de2025unconditional}. By leveraging naturally available multilevel quantum states, qudits offer advantages in information density, noise resilience, and algorithmic efficiency~\cite{ chi2022programmable, kiktenko2025colloquium,meth2025simulating}. For instance, high-dimensional entangled states such as Greenberger-Horne-Zeilinger (GHZ) states and cluster states demonstrate greater robustness against noise compared to their qubit analogs~\cite{reimer2019high}. These benefits have driven growing interest across a range of qudit-based hardware platforms~\cite{karacsony2024efficient,kim2024qudit, nguyen2024empowering}, underscoring the importance of extending fidelity estimation protocols and their resource-theoretic analysis to the qudit regime. 
                  
In the qudit systems with odd prime dimensions, the negativity of the discrete Wigner function quantifies nonstabilizerness resources in states, unitaries, and channels~\cite{Gross_2006, Veitch-2014-NewJ.Phys., Wang-2020-Phys.Rev.Lett.a, WWS19}. Building on this foundation, we introduce a computationally efficient nonstabilizerness measure called the Wigner rank and use it to quantify the complexity of direct fidelity estimation for both quantum states and channels in odd prime dimensions. Furthermore, we reveal that implementing a DFE on general qudit states requires resources that scale exponentially with the nonstabilizerness of the target, as captured by both the mana and the Wigner rank. Notably, the DFE task is feasible for classical-efficient simulable states and quantum evolutions. 
Leveraging the Wigner rank, we establish a tighter upper bound on the sample complexity of DFE than previous results~\cite{Flammia-2011-Phys.Rev.Lett.}, establishing a direct connection between the hardness of fidelity estimation and the underlying nonstabilizerness of the target. In doing so, we extend the findings of prior work~\cite{Leone_2023} from multi-qubit systems to high-dimensional qudit systems, while also offering detailed implementations of the corresponding protocols for these systems.

The rest of this paper is organized as follows. In Sec.~\ref {sec:preliminaries}, we provide key preliminaries to establish the necessary background for introducing the QRT of magic states and channels in systems with odd prime dimensions, including stabilizer formalism~\cite{gottesman1997stabilizer}, the discrete Wigner function~\cite{Gross_2006, Gross2007}, and a magic measure for quantum states, mana~\cite{Veitch-2014-NewJ.Phys.}, and its channel version~\cite{WWS19}. In Sec.~\ref {sec:wignerrank}, we introduce a new magic measure, the Wigner rank, for quantum states and extend it to quantum unitary channels. We prove several desirable properties of the Wigner rank, including faithfulness, additivity, and being an upper bound of mana. We also establish key properties of the Wigner rank for quantum channels, such as faithfulness, additivity for tensor products of channels, and subadditivity for serial compositions of channels. In Sec .~\ref {sec:fidelity-estimation-state-l2} we present the fidelity estimation protocols of quantum states using the Wigner rank. In Sec .~\ref {sec:fidelity-estimation-state-l1}, we introduce the fidelity estimation protocols of quantum states via mana, along with an improved protocol for well-conditioned states, specifically pure stabilizer states. In Sec.~\ref{fidelity-estimation-Chan-norm2} and Sec.~\ref{fidelity-estimation-chan-norm1}, we extend these protocols to the quantum unitary channels. Finally, we conclude in Sec.~\ref{sec:con}.

\section{Preliminaries}\label{sec:preliminaries}
In this section, we briefly introduce several fundamental concepts in the quantum resource theory (QRT) of magic states and channels in systems with odd prime dimensions~\cite{Gross_2006, Veitch-2014-NewJ.Phys., WWS19, Wang-2020-Phys.Rev.Lett.a}.

\subsection{Stabilizer formalism}\label{sec:stabilizer-formalism}

In fault-tolerant quantum computation schemes, the restricted set of quantum operations is the stabilizer operations (SOs). Here, we provide an overview of the basic framework of the stabilizer formalism applicable to quantum systems with odd prime dimensions.

We denote a Hilbert space of dimension $d$ by $\cH_d$ and 
the standard computational basis by $\{\ket j\}_{j=0,\ldots,d-1}$. 
For an odd prime number $d$, we define the respective shift and boost  operators $X,Z\in\cL(\cH_d)$ as
\begin{align}
X\ket j &= \ket{j\oplus 1} , \\
Z\ket j &=\omega^j \ket j, 
\end{align}
where $\oplus$ denotes addition modulo $d$ and $\omega:=e^{2\pi i /d}$.
The Heisenberg--Weyl operators are defined as
\begin{align}
T_{\bu}= \tau^{-a_1a_2}Z^{a_1}X^{a_2},
\end{align}
where $\tau=e^{(d+1)\pi i/d}$ and $\bu=(a_1,a_2)\in \bZ_d\times \bZ_d$.

For a system with composite Hilbert space $\cH_{A}\ox\cH_B$, the Heisenberg--Weyl operators are
the tensor product of the subsystem Heisenberg--Weyl operators:
\begin{align}
T_{\bu_A\oplus \bu_B} = T_{\bu_A} \ox T_{\bu_B},
\end{align} 
where $\bu_A\oplus \bu_B$ is an element of $\bZ_{d_A}\times\bZ_{d_A}\times\bZ_{d_B}\times\bZ_{d_B}$.

The set $\cC_d$ of Clifford operators is defined to be the set of unitary operators that map Heisenberg--Weyl operators to Heisenberg--Weyl operators under conjugation up to phases:
\begin{align}
U\in \cC_d \text{ iff. } \forall \bu, \exists \, \theta,\bu', \text{ such that }
UT_\bu U^{\dagger} = e^{i\theta}T_{\bu'},
\end{align}
where ``iff.'' is short for if and only if.
The Clifford operators form the Clifford group.

The pure stabilizer states can be obtained by applying Clifford operators to the state $\ket 0$:
\begin{align}
\{S_j\}_j=
\{U\proj0U^\dagger: U\in \cC_d \}.
\end{align}
The set of stabilizer states is the convex hull of the set of pure stabilizer states:
\begin{align}
\STAB(\cH_d)=\left\{\rho\in\cS(\cH_d): \rho=\sum_j p_j S_j, \ \forall j \ p_j\ge 0, \ \sum_j p_j = 1\right\}.
\end{align}

SOs for $\rho \in \STAB(\cH_d)$ consists of the following operations~\cite{Veitch-2014-NewJ.Phys., Wang-2020-Phys.Rev.Lett.a}: 
\begin{enumerate}
    \item Clifford operations: $U \rho U^{\dagger}, U \in \mathcal{C}_d$;
    \item Tensoring stabilizer states: $\rho \otimes \rho_s$, where $\rho_s \in \STAB(\cH_d)$, is a stabilizer state;
    \item Partial trace;
    \item Measurements in the computational basis,i.e., on the $i$-th qubit: $(\mathbb{I} \otimes \ketbra{i}{i} \otimes \mathbb{I})\rho (\mathbb{I} \otimes \ketbra{i}{i} \otimes \mathbb{I})$;
    \item Post-processing conditional on the above measurement results.
    
\end{enumerate}

\subsection{Discrete Wigner function}\label{sec:Wigner-function}

The discrete Wigner function \cite{Wootters1987,Gross_2006,Gross2007} has been employed to demonstrate the existence of bound magic states \cite{Veitch_2012}, which are states from which it is impossible to distill magic at a strictly positive rate. 
For an overview of discrete Wigner functions, we refer to \cite{Veitch_2012,Veitch2014,Ferrie2011} for more details.
See also \cite{Ferrie2011} for quasi-probability representations of quantum
theory with applications to quantum information science. 

For each point $\bu\in\bZ_d^n\times\bZ_d^n$ in the discrete phase space, there is a corresponding operator $A_\bu$, 
and the value of the \emph{discrete Wigner representation} of a Hermitian operator $X\in\linear{\cH_d^{\ox n}}$ at this point is given by
\begin{align}
W_X(\bu)=\frac{1}{d^n} \tr[A_\bu X],
\end{align}
where $\{A_\bu\}_\bu$ are the \emph{phase-space point operators} in $\linear{\cH_d^{\ox n}}$:
\begin{align}\label{eq:phase-space point operators}
A_\bu := T_\bu A_\mathbf{0} T_\bu^\dag, \quad \quad A_\mathbf{0} :=\frac{1}{d^n} \sum_\bu T_\bu.
\end{align}
These operators are Hermitian, so the discrete Wigner representation is real-valued. 
There are $d^{2n}$ such operators for $d^n$-dimensional Hilbert space, corresponding to the 
$d^{2n}$ points of discrete phase space.
Some well-known and useful properties of the set $\{A_\bu\}_\bu$ of phase-space point operators and the discrete Wigner representation are listed 
for reference in the following lamma. In Appendix~\ref{appx:lemma:properties}, we prove these properties for self-contained reference.
\begin{lemma}\label{lemma:properties}
The following properties hold regarding the discrete Wigner representation~\cite{WWS19}:
\begin{enumerate}
\item $A_\bu$ is Hermitian;
\item $\tr[A_\bu A_{\bu'}] = d^n\delta(\bu,\bu')$, meaning that $\{A_{\bm{u}}\}_{\bm{u}}$ 
            is an orthogonal basis of $\linear{\cH_d^{\ox n}}$;
\item $\tr [A_\bu] =1$;
\item $\norm{A_{\bm{u}}}{\infty}\leq1$, where $\norm{\cdot}{\infty}$ is the operator norm.
\item $X=\sum_\bu W_X(\bu) A_\bu$ for arbitrary Hermitian operator $X\in\linear{\cH_d^{\ox n}}$. 
\item Let $\rho$ be a quantum state. $\vert W_\rho(\bu)\vert\leq 1/d^n$ for arbitrary $\bu\in\bZ_d^n\times\bZ_d^n$.
\end{enumerate}
\end{lemma}

The inner product of two Hermitian operators can be expressed within the 
discrete Wigner function representation in a compact form, as revealed in the following lemma, 
which is proved in Appendix~\ref{appx:lemma:inner-product}.

\begin{lemma}\label{lemma:inner-product}
Let $M$ and $N$ be two Hermitian operators in $\linear{\cH_d^{\ox n}}$. It holds that
\begin{align}
\tr[MN] = d^n\sum_{\bu} W_M(\bu) W_N(\bu).
\end{align}
Specially, let $\psi$ be a pure state in $\cH_d^{\ox n}$. It holds that
\begin{align}
    d^n\sum_{\bu\in\bZ_d\times \bZ_d}W^2_{\psi}(\bu) = 1.
\end{align}
\end{lemma}

We say a state $\rho$ has positive representation if $W_\rho(\bm{u})\geq0$ for
arbitrary $\bu\in\bZ_d^n\times\bZ_d^n$ and negative representation otherwise. 
There is an important fact about the discrete Wigner representation. 

\begin{theorem}[Discrete Hudson's theorem~\cite{Gross_2006}]\label{thm:Hudson-theorem}
Let $\psi$ be a pure state in $\cH_d^{\ox n}$. It holds that
$\psi$ has positive representation if and only if it is a stabilizer state.
What's more, if $\psi$ is a pure stabilizer state in $\cH_d^{\ox n}$, 
the following statements hold:
\begin{itemize}
    \item $\{W_\psi(\bm{u})\}_{\bm{u}}$ is uniformly distributed, i.e.,
        \begin{align}
            W_\psi(\bm{u}) = 
                \begin{cases}
                    1/d^n,  & \bm{u}\in S_\psi, \\
                    0,      & \bm{u}\not\in S_\psi,
                \end{cases}
        \end{align}
        where $S_\psi$ is the set of stabilizers of $\psi$ (a little abuse definition) with $\vert S_\psi\vert=d^n$.
    \item $\{W_\psi(\bm{u})\}_{\bm{u}}$ is an valid probability distribution:
            \begin{align}\label{eq:probability-distribution}
                1 = \tr[\psi] = \sum_{\bu\in \bZ_d^n\times \bZ_d^n}W_{\psi}(\bu) = \sum_{\bu\in S_\psi}W_{\psi}(\bu).
            \end{align}
\end{itemize}
\end{theorem}

\subsection{Mana}
Let $\cD(\cH_A)$ be the set of density operators acting on $\cH_A$ and denote the set of quantum states that have positive Wigner functions as $\cW_{+}$, i.e., 
\begin{equation}
    \cW_{+} := \{\rho: \rho \in \cD(\cH_d), \forall \ \mathbf{u}, W_{\rho}(\mathbf{u})\geq 0\}.
\end{equation}
A magic monotone called mana~\cite{Veitch2014} of a quantum state $\rho$ is defined as
\begin{equation}
    \cM(\rho) := \log\left(\sum_{\mathbf{u}}|W_{\rho}(\mathbf{u})|\right).
\end{equation}
We have $\cM(\rho)\geq 0$ and $\cM(\rho) = 0$ if and only if $\cM(\rho)\in \cW_{+}$.
As a generalization of the Gottesman-Knill theorem, a quantum circuit where the initial state and all the following quantum operations have positive discrete Wigner functions can be classically simulated~\cite{Mari2012, Veitch_2012}. Thus, for the QRT of magic states in odd prime dimensions, the free states can be chosen as $\cW_{+}$, and the maximal free operations are those that completely preserves the positivity of Wigner functions.

Let $\cN_{A\rightarrow B}$ be a quantum channel mapping system $A$ to $B$ which is a completely positive and trace-preserving (CPTP) map, we introduce the discrete Wigner function of a quantum channel as follows.

\begin{definition}[Discrete Wigner function of a quantum channel~\cite{WWS19}]
    The discrete Wigner function of a given quantum channel $\mathcal{N}_{A \rightarrow B}$ is defined as:
    \begin{align}
    \label{def:DWF_channel}
        \mathcal{W}_{\mathcal{N}}(\mathbf{v}|\mathbf{u}) &\coloneqq \frac{1}{d_B} \tr[((A_A^\mathbf{u})^T)\otimes A_B^\mathbf{v} \mathcal{J}_{AB}^\mathcal{N}] \\
        & = \frac{1}{d_B} \tr(A_B^\mathbf{v} \mathcal{N}(A_A^\mathbf{u})),
    \end{align}
    where phase-space point $\mathbf{u}$ lies in input system A, phase-space point $\mathbf{v}$ lies in output system B, $\mathcal{J}_{AB}^\mathcal{N} = \sum_{ij}\ketbra{i}{j}_A \otimes \mathcal{N}(\ketbra{i}{j}_{B})$ denotes the Choi-Jamio\l kowski matrix.
\end{definition}

The Choi-Jamio\l kowski matrix of a channel $\cN \in \cL(\cH_A^{\ox n}, \cH_{B}^{\ox n})$ has a discrete Wigner representation as follows:
\begin{equation}
    \cJ_{\cN} = \frac{1}{d_A}\sum_{\bm{u},\bm{v}} W_{\cN}(\bm{v}|\bm{u}) A_{\bm{u}}^T \otimes A_{\bm{v}}
\end{equation}

\begin{definition}[CPWP operation~\cite{WWS19}]
A Hermiticity-preserving linear map $\Pi$ is called completely positive Wigner preserving (CPWP) if, for any system $R$ with odd dimension, the following holds 
\begin{equation}
    \forall \rho_{RA}\in \cW_{+},  (\mathrm{id}_R\ox \Pi_{A\rightarrow B}) \rho_{RA} \in \cW_{+},
\end{equation}
\end{definition}
where $\cW_{+}$ denotes the free states of the QRT of magic states in odd prime dimensions.

\begin{definition}[Mana of a quantum channel~\cite{WWS19}]
The mana of a quantum channel $\cN_{A \rightarrow B}$ is defined as:
\begin{align}\label{def:Mana_channel}
    \mathcal{M}(\cN_{A \rightarrow B}) \coloneqq \log \max_\mathbf{u} \| \cN_{A \rightarrow B}(A_A^\mathbf{u})\|_{W,1} = \log \max_\mathbf{u} \sum_\mathbf{v} |W_{\cN}(\mathbf{v}|\mathbf{u})|,
\end{align}
where $\|V \|_{W,1} \coloneqq \sum_\mathbf{u} |W_V(\mathbf{u})| = \sum_\mathbf{u} 1/d|\tr[A_{\mathbf{u}}V]|$ is the Wigner trace norm of a Hermitian operator $V$.     
\end{definition}
It is shown that $\cM(\cN_{A\rightarrow B}) \geq 0$ and $\cM(\cN_{A\rightarrow B}) = 0$ if and only if $\cN_{A\rightarrow B}\in \text{CPWP}$. Also, it is proved that a quantum channel $\cN_{A\rightarrow B}$ is CPWP if and only if the discrete Wigner functions of $J_{\cN}$ are non-negative~\cite{WWS19}.

Note that
\begin{subequations}
    \begin{align}
        \sum_{\bm{v}}\cW_{\cN}(\bm{v}|\bm{u}) &= \sum_{\bm{v}}\frac{1}{d_B} \tr(A_B^\mathbf{v} \cN(A_A^\mathbf{u})) \\
        & = \tr[\sum_{\bm{v}}\frac{A_B^\mathbf{v}}{d_B}\cN(A_A^\mathbf{u})] \\
        & = \tr[\cN(A_A^\mathbf{u})] \\
        & = \tr[A_A^\mathbf{u}] \\
        & = 1.
    \end{align}
\end{subequations}

\section{Wigner rank}
\label{sec:wignerrank}

In this section, we study the \emph{Wigner rank}, a relatively unexplored magic measure for quantum states and quantum channels, inspired by the Pauli rank originally introduced for qubit systems in~\cite{PhysRevLett.123.170502}. The Wigner rank was first proposed in~\cite{bu2024extremality} for quantum states, where it was employed to establish a novel uncertainty principle. Building on this foundation, we extend the definition to quantum channels, conduct a detailed analysis of its mathematical properties, examine its connection to another widely used magic measure known as mana, and ultimately endow it with an operational interpretation through its role in determining the sample complexity of direct fidelity estimation. 


\subsection{Wigner rank of quantum states}

\begin{definition}[Wigner rank of a quantum state]
    Let $\psi$ be a pure state in $\cH_d^{\ox n}$.
We define the \emph{Wigner set} of $\psi$ to be the set of 
phase-space point operators with nonzero Wigner coefficients:
\begin{align}\label{eq:Wigner-non-set} 
    \cW_\psi := \left\{\bm{u}\in\bZ_d^n\times\bZ_d^n 
                        \;\middle\vert\; W_\psi(\bm{u}) \neq 0\right\}.
\end{align}
The \emph{Wigner rank} of $\psi$ 
is defined as the number of nonzero Wigner coefficients of $\psi$, which is the cardinality of the Wigner set:
\begin{align}\label{eq:Wigner-rank}
    \chi(\psi) := \left\vert\cW_\psi\right\vert.
\end{align}
Furthermore, we define the logarithmic Wigner rank as follows:
\begin{align}\label{eq:log-Wigner-rank}
    \chi_{\log}(\psi) := \log\chi(\psi) - \log d(\psi),
\end{align}
where $d(\psi)$ is the underlying dimension of $\psi$.
\end{definition}

Now we show some desirable properties of Wigner rank so that it is indeed an interesting nonstabilizerness measure.

\begin{proposition}[Faithfulness]
Let $\psi$ be a pure state in $\cH_d^{\ox n}$.
It holds that $0 \leq \chi_{\log}(\psi) \leq n\log d$. 
The equality holds if and only if $\ket{\psi}$ is a stabilizer state.
\end{proposition}
\begin{proof}
    $0 \leq \chi_{\log}(\psi) \leq n\log d$ follows directly from the definition.
    According to Theorem~\ref{thm:Hudson-theorem}, we arrive at the conclusion immediately. We also note that this faithfulness property was previously observed in~\cite{bu2024extremality}. 
\end{proof}

\begin{proposition}[Additivity]
Let $\ket{\psi},\ket{\phi}$ be two pure states in $\cH_d^{\ox n}$, it holds that
\begin{align}
    \chi_{\log}(\psi\otimes\phi) = \chi_{\log}(\psi) + \chi_{\log}(\phi).
\end{align}
\end{proposition}
\begin{proof}
    This property also follows directly from the definition.
\end{proof}

\begin{proposition}[Upper bound of Mana]
Let $\psi$ be a pure state in $\cH_d^{\ox n}$, it holds that $\cM(\psi) \leq \chi_{\log}(\psi)$.
\end{proposition}
\begin{proof}
    Recall the definition of the mana of quantum states, we have:
    \begin{equation}
        \cM(\psi) = \log(\sum_{\bm{u}}|W_{\psi}(\bm{u})|) = \log(\sum_{\bm{u}}|\tr[A_{\bm{u}}\psi]|/d^n)
        \leq \log(\sum_{\bm{u} \in \cW_{\psi}}\frac{1}{d^n}) = \chi_{\log}(\psi)
    \end{equation}
    Thus, we conclude that logarithmic Wigner rank is the upper bound of mana for any given pure qudit state with odd prime dimension.
\end{proof}

However, the Wigner rank of mixed stabilizer states does not satisfy the faithfulness condition. A straightforward counterexample is the qutrit maximally mixed state $\mathbb{I}/3$ with Wigner rank being $1$. Thus, Wigner rank can only quantify the non-stabilizerness of pure qudit states, similar to the stabilizer rank~\cite{Bravyi_2016,Bravyi_2019} or stabilizer R\'enyi entropy defined in the multi-qubit scenario. The monotonicity of Wigner rank can be revealed only in subset of SOs transforming pure states into pure (sub-normalized) states, such as Clifford operations, tensoring in pure stabilizer states, where Wigner rank remains invariant under these two operations, and measurements in the computational basis.

A direct generalization of the Wigner rank from pure states to mixed states is given by:
\begin{equation}
    \label{eq:Gen_WR}
    \widetilde{\chi}_{\log}(\rho) = \inf_{\{\ket{\psi_i}, p_i\}} \sum_i p_i \chi_{\log}(\ketbra{\psi_i}{\psi_i}),
\end{equation}
where the infimum is taken over all pure-state decompositions of $\rho$, represented as $\rho = \sum_i p_i \ketbra{\psi_i}{\psi_i}$.

\subsection{Wigner rank of quantum channels}
\begin{definition}
    Let $\cU \in \cL(\cH_A^{\ox n}, \cH_{A'}^{\ox n})$ be a channel corresponding to some unitary evolution, namely $\cU(\rho) = U\rho U^{\dagger}$. 
We define the \emph{Wigner set} of $\cU$ to be the set of 
phase-space point operators with nonzero Wigner coefficients:
\begin{align}\label{eq:Wigner-non-set-unitary}
    \cW_{\cU} := \left\{\bm{v}\in\bZ_d^n\times\bZ_d^n , \bm{u}\in\bZ_d^n\times\bZ_d^n 
                        \;\middle\vert\; W_{\cU}(\bm{v}|\bm{u}) \neq 0\right\}.
\end{align}
The \emph{Wigner rank} of $\cU$ 
is defined as the number of nonzero Wigner coefficients of $\cU$:
\begin{align}\label{eq:Wigner-rank-unitary}
    \chi(\cU) := \left\vert\cW_{\cU}\right\vert.
\end{align}
The normalized logarithm version of $\chi(\psi)$ is defined as
\begin{align}\label{eq:log-Wigner-rank-channel}
    \chi_{\log}(\cU) := \log\chi(\cU) - \log d(\cU),
\end{align}
where $d(\cU)$ is the underlying dimension of $\cU$. For example, if $\cU \in \cL(\cH_A, \cH_{A'})$ with $\text{dim} (\cH_A) = \text{dim} (\cH_{A'}) = d$, $d(\cU) = d^2$.
\end{definition}

Now we show some desirable properties of the Wigner rank of quantum channels.

\begin{proposition}[Faithfulness]
    Let $\cU \in \cL(\cH_A^{\ox n}, \cH_{A'}^{\ox n})$ be a channel corresponding to some unitary evolution.
It holds that $\chi_{\log}(\cU) \geq 0$. 
The equality holds if and only if $\cU$ is a Clifford operation.
\end{proposition}

\begin{proof}
    Note that $\sum_{\bm{v}}\cW_{\cU}(\bm{v}|\bm{u}) = 1$, then for a fixed phase-space point operator $A_{\bm{u}}$, there must exist at least one $A_{\bm{v'}}$ such that $\cW_{\cU}(\bm{v'}|\bm{u}) \neq 0$. This implies that there are at least $d^{2n}$ non-zero Wigner coefficients for a given unitary channel $\cU$ and thus $\chi_{\log}(\cU) \geq 0$.
    A Clifford operation $U_{\bm{F},\bm{a}} \in \cC_{d^n}$ is specified as $U_{\bm{F},\bm{a}} T_{\bm{u}} U_{\bm{F},\bm{a}}^{\dagger}$, where $\bm{u} \in (\bZ_d^n\times\bZ_d^n)$, $\bm{F}$ denotes a $2n \times 2n$ symplectic matrix with entries in $\bZ_d$ and $\bm{a} \in (\bZ_d^n\times\bZ_d^n)$. And Clifford operators satisfy that $U_{\bm{F},\bm{a}} A_{\bm{u}} U_{\bm{F},\bm{a}}^{\dagger} = A_{\bm{F}\bm{u}+\bm{a}}$. For the Wigner function of a Clifford operation $U_{\bm{F},\bm{a}}$, we have:
    \begin{subequations}
          \begin{align}
        \mathcal{W}_{\cU}(\bm{v}|\bm{u})  & = \frac{1}{d^n} \tr(A_\mathbf{v} \cU(A_\mathbf{u}))\\
        &  = \frac{1}{d^n} \tr\big(A_\mathbf{v} (U_{\bm{F},\bm{a}} A_{\bm{u}} U_{\bm{F},\bm{a}}^{\dagger}) \big) \\
        & = \frac{1}{d^n} \tr\big(A_\mathbf{v}  A_{\bm{F}\bm{u}+\bm{a}} \big) \\
         & = \frac{1}{d^n} d^n \delta_{\bm{v}, \bm{F}\bm{u}+\bm{a}} \\
         & = \delta_{\bm{v}, \bm{F}\bm{u}+\bm{a}},
    \end{align}
    \end{subequations}
    We conclude that $\chi(\cU) = d^{2n}$ and $\chi_{\log}(\cU) = 0$ for any Clifford operations. Conversely, if $\chi_{\log}(\cU) = 0$ for a unitary channel $\cU$, then $ \mathcal{W}_{\cU}(\bm{v}|\bm{u}) = 1$ for a given $\bm{u}$. For any pure stabilizer state $\psi$, the discrete Wigner function of $\cU(\psi)$ is:
    \begin{subequations}
        \begin{align}
            \cW_{\cU(\psi)}(\bm{v})& = \frac{1}{d} \tr[A_{\bm{v}} \cU(\psi)] \\
            & =   \frac{1}{d} \sum_{\bm{u}}W_{\psi}(\bm{u})\tr[A_{\bm{v}}\cU(A_{\bm{u}})] \\
            & =  \sum_{\bm{u}}W_{\psi}(\bm{u})\mathcal{W}_{\cU}(\bm{v}|\bm{u}) 
        \end{align}
    \end{subequations}
    Since $\psi$ is a pure stabilizer state, $W_{\psi}(\bm{u})$ takes $0$ or $1/d$, and $\mathcal{W}_{\cU}(\bm{v}|\bm{u})$ takes $0$ or $1$. Therefore, $\cW_{\cU(\psi)}(\bm{v})$ equals either 0 or $1/d$. This implies that $\cU(\psi)$ is always a pure stabilizer state, and hence $\cU$ is a Clifford operation. 
\end{proof}

\begin{proposition}[Additivity]
    For two given unitary channel $\cU_1, \cU_2 $, it holds that 
    \begin{equation}
        \chi_{\log}(\cU_1 \otimes \cU_2) = \chi_{\log}(\cU_1) + \chi_{\log}(\cU_2). 
    \end{equation}
\end{proposition}

\begin{proof}
   Obviously, it holds that:  
    \begin{subequations}
        \begin{align}
        \mathcal{W}_{\cU_1 \otimes \cU_2}(\bm{v}_1 \oplus \bm{v}_2 |\bm{u}_1\oplus \bm{u}_2)  & = \frac{1}{d^{(n_1+n_2)}} \tr\big(A_{\bm{v}_1 \oplus \bm{v}_2} \cU_1 \otimes \cU_2(A_{{\bm{u}_1 \oplus \bm{u}_2} })\big)\\
        & = \frac{1}{d^{n_1}} \tr\big(A_{\bm{v}_1 } \cU_1(A_{{\bm{u}_1} })\big)\frac{1}{d^{n_2}} \tr\big(A_{\bm{v}_2 } \cU_2(A_{{\bm{u}_2} })\big)    \\
        & = \mathcal{W}_{\cU_1}(\bm{v}_1 |\bm{u}_1)\mathcal{W}_{\cU_2}(\bm{v}_2 |\bm{u}_2) ,
        \end{align}
    \end{subequations}
    Then we have $\chi(\cU_1 \otimes \cU_2) = \chi(\cU_1)\cdot \chi(\cU_2)$ by definition and conclude that $\chi_{\log}(\cU_1 \otimes \cU_2) = \chi_{\log}(\cU_1) + \chi_{\log}(\cU_2)$.
\end{proof}

\begin{proposition}[Subadditivity]
    For two given unitary channel $\cU_1, \cU_2 $, it holds that 
    \begin{equation}
        \chi_{\log}(\cU_1 \circ \cU_2) \leq \chi_{\log}(\cU_1) + \chi_{\log}(\cU_2). 
    \end{equation}
\end{proposition}

\begin{proof}
   Recall the properties of the discrete Wigner function. For the composition of two channels, we have: 
   \begin{subequations}
       \begin{align}
           \mathcal{W}_{\cU_1 \circ \cU_2}(\bm{v}|\bm{u}) &=  \frac{1}{d^n} \tr\big(A_{\bm{v}} \cU_2\circ \cU_1(A_{\bm{u} })\big)   \\
           & =  \frac{1}{d^n}\sum_{\bm{w}}  \mathcal{W}_{\cU_1}(\bm{w}|\bm{u})\tr\big(A_{\bm{v}} \cU_2(A_{\bm{w} })\big)   \\
        & = \sum_{\bm{w}} \mathcal{W}_{\cU_2}(\bm{v}|\bm{w}) \mathcal{W}_{\cU_1}(\bm{w}|\bm{u})
       \end{align}
   \end{subequations}
    
   For $\mathcal{W}_{\cU_1 \circ \cU_2}(\bm{v}|\bm{u})$ to be nonzero, at least one term in the sum must be nonzero. This implies that:
   \begin{equation}
       \mathcal{W}_{\cU_2}(\bm{v}|\bm{w}) \neq 0 \ \text{and} \ \mathcal{W}_{\cU_1}(\bm{w}|\bm{u}) \neq 0 \ \text{for some $\bm{w}$}.
   \end{equation}
   For each nonzero entry $(\bm{w},\bm{u})$ of $\mathcal{W}_{\cU_1}(\bm{w}|\bm{u})$, there can be up to $\chi(\cU_2)$ nonzero entries $(\bm{v},\bm{w})$ for $\mathcal{W}_{\cU_2}(\bm{w}|\bm{u})$, each contributing to a nonzero entry $(\bm{v},\bm{u})$ in $\mathcal{W}_{\cU_1 \circ \cU_2}(\bm{v}|\bm{u})$.
   Thus, we obtain the bound:
   \begin{equation}
       \chi(\cU_1 \circ \cU_2) \leq \chi(\cU_1)\cdot \chi(\cU_2),
   \end{equation}
   which leads to the desired inequality: 
   \begin{equation}
       \chi_{\log}(\cU_1 \circ \cU_2) \leq \chi_{\log}(\cU_1) + \chi_{\log}(\cU_2).
   \end{equation}
   Therefore, the logarithmic Wigner rank of quantum channels satisfies subadditivity under serial composition.
\end{proof}


\section{Fidelity estimation of quantum states via Wigner rank}
\label{sec:fidelity-estimation-state-l2}

In the task of fidelity estimation, we are given a quantum device $\mathbf{D}$ 
designed to generate a known pure state $\psi$ in $\cH_d^n$.
However, this device might work incorrectly and generate an unknown
quantum state $\rho$ (possibly mixed) in $\cH_d^n$ each time we call the quantum device.
The fidelity between the desired pure state $\psi$ and the actual state $\rho$ is given by
\begin{align}\label{eq:pure-fidelity}
    \cF(\psi,\rho) 
:= \tr\left[\sqrt{\sqrt{\psi}\rho\sqrt{\psi}}\right]^2
= \tr[\psi\rho].
\end{align}
We need to estimate the fidelity $\cF(\psi,\rho)$, given access to the state copies.
The performance of a reasonable estimation protocol can be quantified by two distinct metrics:
\begin{itemize}
  \item \textbf{Sample complexity:} The number of quantum states is measured in order to 
        collect sufficient data to make the estimation. 
  \item \textbf{Measurement complexity:} The number of different measurement settings 
        in which data are obtained from the measurement device.
  
\end{itemize} 

For fidelity estimation, we employ the direct fidelity estimation (DFE) method. Briefly, DFE works by measuring only a prefixed number of expectation values of phase-space point operators, selected randomly according to an importance weighting rule. In this section, we tackle the fidelity estimation task by utilizing the $\ell_2$ norm first, where Wigner rank emerges as the key magic measure that governs the efficiency of the estimation process, thereby providing Wigner rank with an operational interpretation.

\subsection{Protocol}

The fidelity between our desired pure state $\psi$ and the actual state $\rho$ is given by
\begin{subequations}\label{eq:pure-fidelity-1}
\begin{align}
  \cF(\psi,\rho) 
= \tr[\psi\rho] 
&= d^n\sum_{\bm{u}\in\bZ_d^n\times\bZ_d^n}W_\psi(\bm{u})W_\rho(\bm{u}) \\
&= \sum_{\bm{u}\in\bZ_d^n\times\bZ_d^n} d^nW^2_\psi(\bm{u}) \times \frac{W_\rho(\bm{u})}{W_\psi(\bm{u})} \\
&\equiv \sum_{\bm{u}\in\bZ_d^n\times\bZ_d^n}{\rm Pr}(\bm{u}) X_{\bm{u}},
\end{align}
\end{subequations}
where ${\rm Pr}(\bm{u}):= d^nW^2_\psi(\bm{u})$ 
is a valid probability distribution thanks to Lemma~\ref{lemma:inner-product},
and $X_{\bm{u}}:= W_\rho(\bm{u})/W_\psi(\bm{u})$ is a random variable. Intuitively, Eq.~\eqref{eq:pure-fidelity-1} expresses that the fidelity $\tr[\psi\rho]$ is the expectation
value of the random variable $X_{\bm{u}}$ with respect to the probability distribution ${\rm Pr}(\bm{u})$, thus converting an estimation problem to a sampling problem. We have invoked the $\ell_2$ norm of the discrete Wigner representation of $\psi$. 

The estimation protocol is summarized in Protocol~\ref{protocol:fidelity-estimation-DFE}. The sample complexity of DFE of quantum states via $\ell_2$ norm is quantified by Wigner rank as in Theorem~\ref{theo:state_Wigner_rank}:
\begin{theorem} [Informal]
\label{theo:state_Wigner_rank}
    Given a pure state $\psi \in \cD(\cH_{d})$, an additive error $\varepsilon$ and success probability $1- \delta$, we propose a fidelity estimation protocol via Wigner Rank, the number of phase-space point operators sampled is  $K=\lceil 8/\varepsilon^2\delta\rceil$,
    the sample complexity is $\cO(\frac{1}{\varepsilon^2\delta} + 
\frac{2^{\chi_{\log}(\psi)}}{\varepsilon^2}\ln\frac{1}{\delta})$, where $\chi_{\log}(\psi)$ denotes logarithmic Wigner rank of $\psi$.
\end{theorem}

\begin{proof}
    The proof is demonstrated in the performance analysis of this protocol in the subsection~\ref{Per: wignerrank}.
\end{proof}

\begin{algorithm}[!hbtp]
\caption{Fidelity estimation of quantum states via Wigner rank.}
\label{protocol:fidelity-estimation-DFE}
\begin{algorithmic}[1]
\State{\textbf{Inputs:} \\
        \hskip1em\textbullet~~~The known target quantum state: $\psi$; \\
        \hskip1em\textbullet~~~The unknown prepared quantum state: $\rho$; \\
        \hskip1em\textbullet~~~The fixed additive error $\varepsilon$; \\
        \hskip1em\textbullet~~~The failure probability $\delta$.}
\State{\textbf{Output:} Estimation $\wt{Y}$, such that $\cF(\psi,\rho)$ lies 
        in $[\wt{Y}-\varepsilon, \wt{Y}+\varepsilon]$ with probability larger than $1-\delta$.}
\State{{Compute $K=\lceil 8/\varepsilon^2\delta\rceil$.}
       \Comment{{\color{gray}Number of phase-space point operators to be sampled}}}
\State{Set $\Delta=1/d^n$.}
\For{$k = 1, 2, \cdots, K$}\Comment{{\color{gray}Phase-space point operator sampling procedure}}
\State{Randomly sample a phase-space point operator $A_{\bm{u}_k}$ w.r.t. ${\rm Pr}(\bm{u})=d^nW^2_\psi(\bm{u})$.}
\State{Compute the number of measurement rounds $N_k$ as 
\begin{align}
    {N_k = \left\lceil\frac{8\Delta^2}{KW^2_\psi(\bm{u}_k)\varepsilon^2}\ln\frac{4}{\delta}\right\rceil.}
\end{align}
}
    \For{$j = 1, 2, \cdots, N_k$}\Comment{{\color{gray}Estimation procedure for $X_{\bm{u}}$}}
        \State{Prepare $\rho$ and measure it on the $A_{\bm{u}_k}$ basis.}
        \State{Record the eigenvalue $O_{j\vert k}\in[-\Delta,\Delta]$ of the corresponding outcome.}
    \EndFor
    \State{Compute the following empirical estimator from the experimental data $\{O_{j\vert k}\}_j$:
        \begin{align}
            \widetilde{X}(\bm{u}_k) :=  \frac{1}{N_k}\sum_{j=1}^{N_k}\frac{O_{j\vert k}}{W_\psi(\bm{u}_k)}.
            \qquad{\color{gray}\sim X(\bm{u}_k)=\frac{W_\rho(\bm{u}_k)}{W_\psi(\bm{u}_k)}}
        \end{align}
    }
\EndFor
\State{Compute the following empirical estimator from the data $\{\widetilde{X}(\bm{u}_k)\}_k$:
\begin{align}
    \widetilde{Y} := \frac{1}{K}\sum_{k=1}^K\widetilde{X}(\bm{u}_k).
            \qquad{\color{gray}\sim \tr[\psi\rho]}
\end{align}
}
\State{Output $\widetilde{Y}$ as an unbiased estimator of $\tr[\psi\rho]$.}
\end{algorithmic}
\end{algorithm}

\subsection{Performance analysis}
\label{Per: wignerrank}
Protocol~\ref{protocol:fidelity-estimation-DFE} is a two-stage estimation procedure: 
we first estimate each $\wt{X}_{\bm{u}_k}$ and then use them to estimate the target $\widetilde{Y}$:

\begin{align}
    O_{j\vert k} \quad\rightarrow\quad \wt{X}_{\bm{u}_k}
\quad\rightarrow\quad \wt{Y} 
\quad\xrightarrow{\text{\color{klevinblue}~Hoeffding inequality~}}\quad Y 
\quad\xrightarrow{ \text{\color{klevinblue}~Chebyshev inequality~}}\quad \tr[\psi\rho].
\end{align}

\paragraph*{Step 1: Estimating $Y$.}
Assume that each $X_{\bm{u}}$ (equivalently, $W_\rho(\bm{u})$) can be estimated \textit{ideally}. To estimate $\tr[\psi\rho]$ with a fixed additive error of $\varepsilon/2$ and significance level (aka. failure probability) $\delta/2$. The reason that we use $\varepsilon/2$ and $\delta/2$ instead of $\varepsilon$ and $\delta$ will be clear later. We perform $K$ samplings, obtaining random variables $X_{\bm{u}_1},\cdots,X_{\bm{u}_K}$. The sample average is then defined as $Y=1/K\sum_{k=1}^KX_{\bm{u}_k}$. We can apply the Chebyshev inequality to bound the estimation error of $Y$.

\textbf{Step 1.1: The expected value $\bE[Y]$ is finite.} 
This is trivial since
\begin{align}
  \bE[Y]  = \bE\left[\frac{1}{K}\sum_{k=1}^KX_{\bm{u}_k}\right] 
          = \sum_{\bm{u}\in\bZ_d^n\times\bZ_d^n} d^nW^2_\psi(\bm{u}) \times \frac{W_\rho(\bm{u})}{W_\psi(\bm{u})}
          = \tr[\psi\rho].
\end{align}
where $\bE$ denotes the expected value over the random choice of $k$.

\textbf{Step 1.2: The variance of $Y$ is bounded.} Since $Y$ is a weighted sum of $X_{\bm{u}_k}$, 
we first compute the variance of $X_{\bm{u}_k}$:
\begin{align}
  \opn{Var}[X_{\bm{u}_k}]
&= \bE[X_{\bm{u}_k}^2] - (\bE[X_{\bm{u}_k}])^2 \\
&= \sum_{\bm{u}\in\bZ_d^n\times\bZ_d^n} d^nW^2_\psi(\bm{u})\left[\frac{W_\rho(\bm{u})}{W_\psi(\bm{u})}\right]^2
    - \left(\tr[\psi\rho]\right)^2 \\
&= \sum_{\bm{u}\in\bZ_d^n\times\bZ_d^n} d^nW^2_\rho(\bm{u}) - \tr[\psi\rho]^2 \\
&= \tr[\rho^2] - \tr[\psi\rho]^2 \\
&\leq \tr[\rho^2] \\
&\leq 1,
\end{align}
where the fourth equality follows from Lemma~\ref{lemma:inner-product}.
Then, the variance of $Y$ follows directly
\begin{align}
    \opn{Var}[Y] = \frac{1}{K^2}\sum_{k=1}^K\opn{Var}[X_{\bm{u}_k}] \leq \frac{1}{K}.
\end{align}

\textbf{Step 1.3: Chebyshev inequality.} We have 
\begin{align}
\opn{Pr}\left\{\vert Y - \tr[\psi\rho]\vert \geq\frac{a}{\sqrt{K}}\right\}
\leq \opn{Pr}\left\{\vert Y - \tr[\psi\rho]\vert \geq a\sqrt{\opn{Var}[Y]}\right\} \leq \frac{1}{a^2}.
\end{align}
Let $a=\sqrt{2/\delta}$ and $K=\lceil 8/\varepsilon^2\delta\rceil$ yields
\begin{align}\label{eq:wignerrank-estimation-bound-1}
    \opn{Pr}\left\{\vert Y - \tr[\psi\rho]\vert \geq \varepsilon/2\right\} \leq \delta/2.
\end{align}

\paragraph*{Step 2: Estimating $X_{\bm{u}_k}$.}
However, each $X_{\bm{u}_k}$ can not be ideally estimated. Therefore, we utilize a finite-precision estimator $\wt{Y}$, derived from a finite number of copies of the state $\rho$, to approximate the ideal infinite-precision estimator $Y$:
\begin{align}
    \wt{Y} 
&= \frac{1}{K}\sum_{k=1}^K\wt{X}_{\bm{u}_k} \\
&= \sum_{k=1}^K\sum_{j=1}^{N_k}\frac{1}{K}\times\frac{1}{N_k}\times\frac{O_{j\vert k}}{W_\psi(\bm{u}_k)}  \\
&= \sum_{k=1}^K\sum_{j=1}^{N_k}\wt{O}_{j\vert k},
\end{align}
where
\begin{align}
  \wt{O}_{j\vert k} := \frac{1}{KN_kW_\psi(\bm{u}_k)}O_{j\vert k} 
\in [-\Delta/(KN_kW_\psi(\bm{u}_k)), \Delta/(KN_kW_\psi(\bm{u}_k))],
\end{align}
where the range bound follows since $O_{j\vert k}\in[-\Delta,\Delta]$, 
where $\Delta=1/d^n$, proven in Lemma~\ref{lemma:properties}. 
Then we apply Hoeffding's inequality:
\begin{subequations}\label{eq:wignerrank-estimation-bound-2}
\begin{align}
    \opn{Pr}\left(\vert\wt{Y} - Y \vert \geq \varepsilon/2\right)
&\leq 2\exp\left(- \frac{2\varepsilon^2}{4\sum_{k=1}^K\sum_{j=1}^{N_k}
                    \frac{4\Delta^2}{K^2N_k^2W^2_\psi(\bm{u}_k)}}\right) \\
&= {2\exp\left(- \frac{\varepsilon^2}{\sum_{k=1}^K
                        \frac{8\Delta^2}{K^2N_kW^2_\psi(\bm{u}_k)}}\right)}.
\end{align}
\end{subequations}
We need to cleverly choose $N_k$ so that the error probability satisfies
\begin{subequations}
\begin{align}
&\quad    {2\exp\left(- \frac{\varepsilon^2}
            {\sum_{k=1}^K\frac{8\Delta^2}{K^2N_kW^2_\psi(\bm{u}_k)}}\right)} = \delta/2 \\
\Leftarrow&\quad \sum_{k=1}^K\frac{8\Delta^2}{K^2N_kW^2_\psi(\bm{u}_k)} 
                = \frac{\varepsilon^2}{\ln\frac{4}{\delta}} \\
\Leftarrow&\quad  \forall k=1,\cdots, K,\quad 
                  \frac{8\Delta^2}{K^2N_kW^2_\psi(\bm{u}_k)} 
                = \frac{1}{K}\times\frac{\varepsilon^2}{\ln\frac{4}{\delta}} \\
\Leftarrow&\quad  \forall k=1,\cdots, K,\quad 
    N_k = \left\lceil\frac{8\Delta^2}{KW^2_\psi(\bm{u}_k)\varepsilon^2}\ln\frac{4}{\delta}\right\rceil.
\end{align}
\end{subequations}

\paragraph*{Step 3: Union bound.}

From the above two steps, we obtain the following two probability bounds given in Eqs.~\eqref{eq:wignerrank-estimation-bound-1} and~\eqref{eq:wignerrank-estimation-bound-2}:
\begin{align}
  \Pr\left\{\vert Y - \tr[\psi\rho]\vert \geq \varepsilon/2\right\}
&\leq \delta/2, \\
\opn{Pr}\left(\vert\wt{Y} - Y \vert \geq \varepsilon/2\right)
&\leq \delta/2.
\end{align}
Using the union bound, we have 
\begin{align}
  \Pr\left\{\vert \wt{Y} - \tr[\psi\rho]\vert \leq \varepsilon\right\}
&= 1 - \Pr\left\{\vert \wt{Y} - \bE[Y]\vert \geq \varepsilon\right\}  \\
&\geq 1 - \left( \Pr\left\{\vert Y - \bE[Y]\vert \geq \varepsilon/2\right\}
               + \opn{Pr}\left(\vert\wt{Y} - Y \vert \geq \varepsilon/2\right)\right) \\
&\geq 1 - \delta.
\end{align}
We can then conclude that,
\begin{quote}
\textit{\color{googlered}
With probability $\geq 1-\delta$, 
the fidelity $\cF(\psi,\rho)$ lies in the range $[\wt{Y}-\varepsilon, \wt{Y}+\varepsilon]$.}
\end{quote}

\paragraph*{Number of copies}

Notice that the DFE method first samples \textbf{$K=\lceil 8/\varepsilon^2\delta\rceil$}
number of phase-space point operators and use 
$N_k=\lceil8\Delta^2/\big(KW^2_\psi(\bm{u}_k)\varepsilon^2\big)\ln(4/\delta)\rceil$
number of state copies to estimate the corresponding expectation value.
The expected value of $N_k$ w.r.t. $k$ is given by
\begin{align}
    \bE[N_k] 
&= \sum_{\bm{u}\in\bZ_d^n\times\bZ_d^n} {\rm Pr}(\bm{u})N_{\bm{u}} \\
&\leq 1 + \sum_{\bm{u}\in\bZ_d^n\times\bZ_d^n} d^nW^2_\psi(\bm{u}) \times 
        \frac{8\Delta^2}{KW^2_\psi(\bm{u})\varepsilon^2}\ln\frac{4}{\delta} \\
&\leq 1 + \frac{8\chi(\psi)}{K\varepsilon^2d^n}\ln\frac{4}{\delta}\\
&\leq 1 + \frac{8\cdot2^{\chi_{\log}(\psi)}}{K\varepsilon^2}\ln\frac{4}{\delta},
\end{align}
where the last inequality follows from the fact that $\Delta=1/d^n$ and the definition of $\chi_{\log}(\psi)$. 
Thus, the total number of samples consumed is given by
\begin{align}
    \bE[N] = \sum_{k=1}^K \bE[N_k]
\leq 1 + \frac{8}{\varepsilon^2\delta} + 
8\cdot2^{\chi_{\log}(\psi)}\times\frac{1} {\varepsilon^2}\ln\frac{4}{\delta}.
\end{align}
By Markov's inequality, $N$ is unlikely to exceed its expectation by much.
In comparison to the original DFE sample complexity for an 
$n$-qubit system, where $d=2^n$~\cite{Flammia-2011-Phys.Rev.Lett.}, we use the magic measure, Wigner rank, to upper bound the sample complexity, instead of relying on the system dimension $d$. This approach not only provides a tighter upper bound, but also establishes a direct connection between the hardness of the fidelity estimation task of an estimated state and its nonstabilizerness resources. 

\paragraph*{Summary of the performance of Protocol~\ref{protocol:fidelity-estimation-DFE}}

\begin{itemize}
  \item \textbf{Sample complexity:} 
    $\cO(\frac{1}{\varepsilon^2\delta} + \frac{2^{\chi_{\log}(\psi)}}{\varepsilon^2}\ln\frac{1}{\delta} )$.  
  \item \textbf{Measurement complexity:} $\cO(d^{2n})$; All phase-space point operators in $\bZ_d^n\times\bZ_d^n$.
  \item \textbf{Measurement type:} single-qudit phase-space point measurements are enough.
\end{itemize}

\section{Fidelity estimation of quantum states via mana}\label{sec:fidelity-estimation-state-l1}

In Section~\ref{sec:fidelity-estimation-state-l1}, we address the fidelity estimation task 
by employing the  $\ell_1$ norm of the discrete Wigner function, where the mana serves as the key magic measure governing the estimation efficiency, which also endows mana with an operational interpretation in the fidelity estimation task.

\subsection{Protocol} 

Let $\Delta_\psi:= 2^{\cM(\psi)} = \sum_{\bm{u}\in\bZ_d^n\times\bZ_d^n}\vert W_\psi(\bm{u}) \vert$ 
be the exponential version of the mana measure~\cite[Definition 7]{Veitch-2014-NewJ.Phys.}. Using this, the fidelity between our desired pure state $\psi$ and the actual state $\rho$ is given by
\begin{subequations}\label{eq:pure-fidelity-2}
\begin{align}
  \cF(\psi,\rho) 
= \tr[\psi\rho] 
&= d^n\sum_{\bm{u}\in\bZ_d^n\times\bZ_d^n}W_\psi(\bm{u})W_\rho(\bm{u}) \\
&= \sum_{\bm{u}\in\bZ_d^n\times\bZ_d^n}\frac{\vert W_\psi(\bm{u}) \vert}{\Delta_\psi}
        \operatorname{sgn}\big(W_\psi(\bm{u})\big)d^n\Delta_\psi W_\rho(\bm{u}) \\
&\equiv \sum_{\bm{u}\in\bZ_d^n\times\bZ_d^n}{\rm Pr}(\bm{u}) X_{\bm{u}},
\end{align}
\end{subequations}
where the second equality follows from Lemma~\ref{lemma:inner-product}, 
$\operatorname{sgn}:\mathbb{R}\to\{1,-1\}$ is the sign function,
${\rm Pr}(\bm{u}):= \vert W_\psi(\bm{u}) \vert/\Delta_\psi$ is a valid probability distribution, and
\begin{align}
    X_{\bm{u}} := \operatorname{sgn}\big(W_\psi(\bm{u})\big)d^n\Delta_\psi W_\rho(\bm{u})
\end{align}
is a random variable. Similar to previous section, Eq.~\eqref{eq:pure-fidelity-2} shows that the fidelity $\tr[\psi\rho]$ is the expectation
value of the random variable $X_{\bm{u}}$ with respect to the probability distribution ${\rm Pr}(\bm{u})$, thereby transforming the estimation problem to a sampling problem. 

 \begin{algorithm}[!hbtp]
\caption{Fidelity estimation of quantum states via mana.}
\label{protocol:fidelity-estimation-DFE-mana}
\begin{algorithmic}[1]
\State{\textbf{Inputs:} \\
        \hskip1em\textbullet~~~The known target quantum state: $\psi$; \\
        \hskip1em\textbullet~~~The unknown prepared quantum state: $\rho$; \\
        \hskip1em\textbullet~~~The fixed additive error $\varepsilon$; \\
        \hskip1em\textbullet~~~The failure probability $\delta$.}
\State{\textbf{Output:} Estimation $\wt{Y}$, such that $\cF(\psi,\rho)$ lies 
        in $[\wt{Y}-\varepsilon, \wt{Y}+\varepsilon]$ with probability larger than $1-\delta$.}
\State{Compute $K=\lceil 8\Delta_\psi/\varepsilon^2\delta)\rceil$.
       \Comment{{\color{gray}Number of phase-space point operators to be sampled}}}
\State{Set $\Delta=1/d^n$.}
\For{$k = 1, 2, \cdots, K$}\Comment{{\color{gray}Phase-space point operator sampling procedure}}
\State{Randomly sample a phase-space point operator $A_{\bm{u}_k}$ 
        w.r.t. ${\rm Pr}(\bm{u})=\vert W_\psi(\bm{u}) \vert/\Delta_\psi$.}
\State{Compute the number of measurement rounds $N_k$ as 
\begin{align}
    {N_k = \left\lceil\frac{8\Delta_{\psi}^2}{K\varepsilon^2}\ln\frac{4}{\delta}\right\rceil.}
\end{align}
}
    \For{$j = 1, 2, \cdots, N_k$}\Comment{{\color{gray}Estimation procedure for $X_{\bm{u}}$}}
        \State{Prepare $\rho$ and measure it on the $A_{\bm{u}_k}$ basis.}
        \State{Record the eigenvalue $O_{j\vert k}\in[-\Delta,\Delta]$ of the corresponding outcome.}
    \EndFor
    \State{Compute the following empirical estimator from the experimental data $\{O_{j\vert k}\}_j$:
        \begin{align}
            \widetilde{X}(\bm{u}_k) := 
                \frac{1}{N_k}\sum_{j=1}^{N_k}\operatorname{sgn}\big(W_\psi(\bm{u})\big)d^n\Delta_\psi O_{j\vert k}.
                    \qquad{\color{gray}\sim X(\bm{u}_k)}
        \end{align}
    }
\EndFor
\State{Compute the following empirical estimator from the data $\{\widetilde{X}(\bm{u}_k)\}_k$:
\begin{align}
    \widetilde{Y} := \frac{1}{K}\sum_{k=1}^K\widetilde{X}(\bm{u}_k).
            \qquad{\color{gray}\sim \cF(\psi,\rho)}
\end{align}
}
\State{Output $\widetilde{Y}$ as an unbiased estimator of $\cF(\psi,\rho)$.}
\end{algorithmic}
\end{algorithm}

The estimation protocol is summarized in Protocol~\ref{protocol:fidelity-estimation-DFE-mana}. We demonstrate that DFE of quantum states via the $\ell_1$ norm is quantified by mana in Theorem~\ref{theo:state_mana}:
\begin{theorem} [Informal]
\label{theo:state_mana}
    Given a pure state $\psi \in \cD(\cH_{d})$, an additive error $\varepsilon$, and success probability $1- \delta$, we propose a fidelity estimation protocol based on mana. In this protocol, the number of phase-space point operators sampled is given by: $K=\lceil 8\cdot 2^{\cM(\psi)}/\varepsilon^2\delta)\rceil$,
    and the sample complexity is $\cO(\frac{2^{\cM(\psi)}}{\varepsilon^2\delta} 
        + \frac{2^{2\cdot\cM(\psi)}}{\varepsilon^2}\ln\frac{1}{\delta})$. 
\end{theorem}

\begin{proof}
    The proof is demonstrated in the performance analysis of this protocol in the Appendix~\ref{appx:theo:state_mana}.
\end{proof}

\paragraph*{Summary of the performance of Protocol~\ref{protocol:fidelity-estimation-DFE-mana}}

\begin{itemize}
  \item \textbf{Sample complexity:} 
  $\cO(\frac{\Delta_\psi}{\varepsilon^2\delta} 
        + \frac{\Delta_\psi^2}{\varepsilon^2}\ln\frac{1}{\delta})$, where $\Delta_\psi = 2^{\cM(\psi)}$.   
  \item \textbf{Measurement complexity:} $\cO(d^{2n})$; All phase-space point operators in $\bZ_d^n\times\bZ_d^n$. 
  \item \textbf{Measurement type:} single-qudit phase-space point measurements are sufficient. 
\end{itemize}

For the case of pure stabilizer states, it holds that $\Delta_{\psi} =  1$, and thus the sample complexity deduces to $\cO({\frac{1}{\varepsilon^2\delta} + \frac{1}{\varepsilon^2}\ln\frac{1}{\delta}})$, which is consistent with the performance of Protocol~\ref{protocol:fidelity-estimation-DFE}.


\subsection{Well-conditioned states}

If we concentrate on estimating the fidelity of pure stabilizer states, we can achieve exponential speedup (compared to the above standard estimation protocol).

\paragraph{Protocol}

This protocol relies crucially on the pure stabilizer state characterization given in Theorem~\ref{thm:Hudson-theorem}. Assume $\cS_\psi$ be the set of phase-space point operators with non-zero discrete Wigner representation of a stabilizer state $\psi$. Then $\vert\cS_\psi\vert=d^n$. 
The fidelity between our desired pure stabilizer state $\psi$ and the actual state $\rho$ can be rewritten as
\begin{align}\label{eq:pure-fidelity-3}
  \cF(\psi,\rho) 
= \tr[\psi\rho] 
= d^n\sum_{\bm{u}\in\cS_\psi}W_\psi(\bm{u})W_\rho(\bm{u})
= \sum_{\bm{u}\in\cS_\psi}W_\psi(\bm{u})\wh{W}_\rho(\bm{u}),
\end{align}
where $\wh{W}_\rho(\bm{u}) := d^nW_\rho(\bm{u})$.
Thank to Theorem~\ref{thm:Hudson-theorem} (see also Eq.~\eqref{eq:pure-fidelity-2}),
we know that $\{W_\psi(\bm{u})\}_{\bm{u}}$ is an valid probability distribution; 
Actually, it is a uniform distribution.
Similarly, we treat this estimation problem as a sampling problem.
The estimation protocol is summarized in Protocol~\ref{protocol:stabilizer-fidelity-estimation-DFE}, whose corresponding cost is introduced in Proposition~\ref{pro:stab}.

\begin{algorithm}[!hbtp]
\caption{Fidelity estimation of pure stabilizer states.}
\label{protocol:stabilizer-fidelity-estimation-DFE}
\begin{algorithmic}[1]
\State{\textbf{Inputs:} \\
        \hskip1em\textbullet~~~The known target pure stabilizer quantum state: $\psi$; \\
        \hskip1em\textbullet~~~The unknown prepared quantum state: $\rho$; \\
        \hskip1em\textbullet~~~The fixed additive error $\varepsilon$; \\
        \hskip1em\textbullet~~~The failure probability $\delta$.}
\State{\textbf{Output:} Estimation $\wt{Y}$, such that $\cF(\psi,\rho)$ lies 
        in $[\wt{Y}-\varepsilon, \wt{Y}+\varepsilon]$ with probability larger than $1-\delta$.}
\State{{Compute $K=\lceil 8/\varepsilon^2\ln(4/\delta)\rceil$.}
       \Comment{{\color{gray}Number of phase-space point operators to be sampled}}}
\For{$k = 1, 2, \cdots, K$}\Comment{{\color{gray}Phase-space point operator sampling procedure}}
\State{Randomly sample a phase-space point operator $\bm{u}_k$ w.r.t. $\{W_\psi(\bm{u}):\bm{u}\in\bZ_d^n\times\bZ_d^n\}$.}
\State{Compute the number of measurement rounds $N_k$ as 
\begin{align}
    {N_k = \left\lceil\frac{8}{K\varepsilon^2}\ln\frac{4}{\delta}\right\rceil = 1.}
\end{align}
}
    \For{$j = 1, 2, \cdots, N_k$}\Comment{{\color{gray}Estimation procedure for $\tr[A_{\bm{u}}\rho]$}}
        \State{Prepare $\rho$ and measure it on the $A_{\bm{u}_k}$ basis.}
        \State{Record the eigenvalue $O_{j\vert k}\in[-1,1]$ of the corresponding outcome.
                Note that $O_{j\vert k}$ is rescaled by $d^n$.}
    \EndFor
    \State{Compute the following emperical estimator from the experimental data $\{O_{j\vert k}\}_j$:
        \begin{align}
            \widetilde{W}_\rho(\bm{u}_k) :=  \frac{1}{N_k}\sum_{j=1}^{N_k}O_{j\vert k}.
            \qquad{\color{gray}\sim \wh{W}_\rho(\bm{u}_k) = 2^n\tr[A_{\bm{u}_k}\rho]}
        \end{align}
    }
\EndFor
\State{Compute the following emperical estimator from the data $\{\widetilde{W}_\rho(\bm{u}_k)\}_k$:
\begin{align}
    \widetilde{Y} := \frac{1}{K}\sum_{k=1}^K\widetilde{W}_\rho(\bm{u}_k).
            \qquad{\color{gray}\sim \tr[\psi\rho]}
\end{align}
}
\State{Output $\widetilde{Y}$ as an unbiased estimator of $\tr[\psi\rho]$.}
\end{algorithmic}
\end{algorithm}

\begin{proposition}[Informal]
    \label{pro:stab}
    Given a pure stabilizer state $\psi \in \cD(\cH_{d})$, an additive error $\varepsilon$ and success probability $1- \delta$, we propose an efficient fidelity estimation protocol, the number of phase-space point operators sampled is $K=\lceil 8/\varepsilon^2\ln\frac{4}{\delta}\rceil$, the sample complexity is $\cO\Big(
       \frac{1}{\varepsilon^2}\ln\frac{1}{\delta}\Big)$.
\end{proposition} 

\begin{proof}
    The proof can be found in the performance analysis of this protocol in the Appendix~\ref{appx:pro:stab}.
\end{proof} 
Proposition~\ref{pro:stab} demonstrates that both the number of phase-space point operators and the sample complexity are independent of the system size for stabilizer states. This yields a more efficient performance compared to both Protocol~\ref{protocol:fidelity-estimation-DFE} and Protocol~\ref{protocol:fidelity-estimation-DFE-mana}, whose sample complexities scale as $\cO\Big(
       \frac{1}{\varepsilon^2\delta}+\frac{1}{\varepsilon^2}\ln\frac{1}{\delta}\Big)$.

\subparagraph*{Summary of the performance of Protocol~\ref{protocol:stabilizer-fidelity-estimation-DFE}}

\begin{itemize}
  \item \textbf{Sample complexity:} $\cO\Big(1/\varepsilon^2\ln(1/\delta)\Big)$. 
  \item \textbf{Measurement complexity:} $\cO(d^n)$; All stabilizing phase-space point operators of $\psi$.
  \item \textbf{Measurement type:} single-qudit phase-space point measurements are enough.
\end{itemize}


\section{Fidelity estimation of quantum channels via Wigner rank}
\label{fidelity-estimation-Chan-norm2}

In section~\ref{fidelity-estimation-Chan-norm2}, we extend the fidelity estimation task of quantum states via Wigner rank into unitary quantum channels.

\subsection{Protocol}

The \emph{entanglement fidelity} between our desired unitary $\cU$ and the actual channel $\Lambda$ is given by
\begin{subequations}\label{eq:chan-fidelity-2}
\begin{align}
  \cF(\cU,\Lambda) 
&= \frac{1}{d^{2n}}\tr[\cU^{\dagger}\Lambda] \\
&= \frac{1}{d^{2n}}\sum_{\bm{u},\bm{v}}W_{\cU}(\bm{v}|\bm{u})W_{\Lambda}(\bm{v}|\bm{u}) \\
&= \sum_{\bm{u},\bm{v}} \frac{W^2_{\cU}(\bm{v}|\bm{u}) }{d^{2n}}\times \frac{W_{\Lambda}(\bm{v}|\bm{u})}{W_{\cU}(\bm{v}|\bm{u})} \\
&\equiv \sum_{\bm{u},\bm{v}}{\rm Pr}(\bm{v},\bm{u}) X_{\bm{v},\bm{u}},
\end{align}
\end{subequations}
where $\bm{v},\bm{u} \in \bZ_d^n\times\bZ_d^n$, and $\tr[\cU^{\dagger}\Lambda]$ is the Hilbert-Schmidt inner product between $\cU$ and $\Lambda$. Similar to the case of state fidelity estimation, ${\rm Pr}(\bm{v},\bm{u})$ denotes an valid probability distribution and $X_{\bm{v},\bm{u}}$ denotes a random variable with ${\rm Pr}(\bm{v},\bm{u}) = \frac{W^2_{\cU} (\bm{v}|\bm{u}) }{d^{2n}}$ and $X_{\bm{v},\bm{u}} = \frac{W_{\Lambda}(\bm{v}|\bm{u})}{W_{\cU}(\bm{v}|\bm{u})}$.  

Similar to the case of states, we convert this estimation problem to a sampling problem. The estimation protocol is summarized in Protocol~\ref{protocol:fidelity-estimation-Chan-norm2}. We show that the DFE of unitary quantum channels via the $\ell_2$ norm is quantified by the Wigner rank of a quantum channel in Theorem~\ref{theo:channel_Wigner_rank}:

\begin{theorem} [Informal]
\label{theo:channel_Wigner_rank}
    Given a unitary channel $\cU \in  \cL(\cH_A^{\ox n}, \cH_{A'}^{\ox n})$, an additive error $\varepsilon$, and success probability $1- \delta$, we
    propose a fidelity estimation protocol via Wigner Rank, the number of phase-space point operators sampled is  $K=\lceil 8/\varepsilon^2\delta\rceil$,
    the sample complexity is $\cO \big(1/(\varepsilon^2\delta) + \frac{2^{\chi_{\log}(\cU)}}{\varepsilon^2}\ln\frac{1}{\delta}\big)$, where $\chi_{\log}(\cU)$ denotes logarithmic Wigner rank of $\cU$.
\end{theorem}

\begin{proof}
    The proof remains in the performance analysis of this protocol in the Appendix~\ref{appx:theo:channel_Wigner_rank}.
\end{proof} 

\begin{algorithm}[!hbtp] 
\caption{Entanglement Fidelity estimation of channels via Wigner rank.}
\label{protocol:fidelity-estimation-Chan-norm2}
\begin{algorithmic}[1]
\State{\textbf{Inputs:} \\
        \hskip1em\textbullet~~~The known target channel: $\cU$; \\
        \hskip1em\textbullet~~~The unknown prepared channel: $\Lambda$; \\
        \hskip1em\textbullet~~~The fixed additive error $\varepsilon$; \\
        \hskip1em\textbullet~~~The failure probability $\delta$.}
\State{\textbf{Output:} Estimation $\wt{Y}$, such that $\cF(\cU,\Lambda)$ lies 
        in $[\wt{Y}-\varepsilon, \wt{Y}+\varepsilon]$ with probability larger than $1-\delta$.}
\State{{Compute $K=\lceil 8/\varepsilon^2\delta)\rceil$.} 
       \Comment{{\color{gray}Number of phase-space point operators to be sampled}}}
\State{Set $\Delta=1$.}
\For{$k = 1, 2, \cdots, K$}\Comment{{\color{gray}Phase-space point operator sampling procedure}}     
\State{Randomly sample a pair of phase-space point operators $(A_{\bm{v}_k}, A_{\bm{u}_k})$ w.r.t. ${\rm Pr}(\bm{v}|\bm{u})=\frac{W^2_{\cU}(\bm{v},\bm{u})}{d^{2n}}$.}
\State{Compute the number of measurement rounds $N_k$ as 
\begin{align}
    {N_k = \left\lceil\frac{8\Delta}{KW^2_{\cU}(\bm{v}_k,\bm{u}_k)\varepsilon^2}\ln\frac{4}{\delta}\right\rceil.}
\end{align}
}
    \For{$j = 1, 2, \cdots, N_k$}\Comment{{\color{gray}Estimation procedure for $X_{\bm{u}}$}}
        \State{Prepare unitary channel $\cU$, apply it on the $A_{\bm{u}_k}$ and measure the outcome on the $A_{\bm{v}_k}$ basis.}
        \State{Record the eigenvalue $O_{j\vert k}\in[-\Delta,\Delta]$ of the corresponding outcome.}
    \EndFor
    \State{Compute the following empirical estimator from the experimental data $\{O_{j\vert k}\}_j$:
        \begin{align}
            \widetilde{X}(\bm{v}_k,\bm{u}_k) :=  \frac{1}{N_k}\sum_{j=1}^{N_k}\frac{O_{j\vert k}}{W_{\cU}(\bm{v}_k|\bm{u}_k)}.
            \qquad{\color{gray}\sim X(\bm{v}_k,\bm{u}_k)=\frac{W_{\Lambda}(\bm{v}_k|\bm{u}_k)}{W_{\cU}(\bm{v}_k|\bm{u}_k)}}
        \end{align}
    }
\EndFor
\State{Compute the following empirical estimator from the data $\{\widetilde{X}(\bm{v}_k,\bm{u}_k)\}_k$:
\begin{align}
    \widetilde{Y} := \frac{1}{K}\sum_{k=1}^K\widetilde{X}(\bm{v}_k,\bm{u}_k).
            \qquad{\color{gray}\sim \frac{\tr[\cU^{\dagger}\Lambda]}{d^{2n}}}
\end{align}
}
\State{Output $\widetilde{Y}$ as an unbiased estimator of $\cF(\cU,\Lambda) = \frac{\tr[\cU^{\dagger}\Lambda]}{d^{2n}}$.}
\end{algorithmic}
\end{algorithm}

\paragraph*{Summary of the performance of Protocol~\ref{protocol:fidelity-estimation-Chan-norm2}}

\begin{itemize}
  \item \textbf{Sample complexity:} 
    $ \cO\big(1/(\varepsilon^2\delta) + \frac{2^{\chi_{\log}(\cU)}}{\varepsilon^2}\ln\frac{1}{\delta}\big)$, where $\chi_{\log}(\cU)$. (In the worst case, $\chi(\cU) = d^{4n}$)
  \item \textbf{Measurement complexity:} $\cO(d^{4n})$; All phase-space point operators in $\bZ_d^n\times\bZ_d^n$.
  \item \textbf{Measurement type:} single-qudit phase-space point measurements are enough.
\end{itemize}

For Clifford operations, $\chi(\cU) =  d^{2n}$, the sample complexity deduces to $\cO\big(1/(\varepsilon^2\delta) + \frac{1}{\varepsilon^2}\ln\frac{1}{\delta}\big)$. 

\section{Fidelity estimation of quantum channels via mana}\label{fidelity-estimation-chan-norm1}

In section~\ref{fidelity-estimation-chan-norm1}, we extend the fidelity estimation task of quantum states via mana into unitary quantum channels.

\subsection{Protocol} 
 
The fidelity between our desired unitary $\cU$ and the actual channel $\Lambda$ is given by
\begin{subequations}\label{eq:chan-fidelity-1}
\begin{align}
  \cF(\cU,\Lambda) 
&= \frac{1}{d^{2n}}\tr[\cU^{\dagger}\Lambda] \\
&= \frac{1}{d^{4n}}\sum_{\bm{v}_1,\bm{u}_1,\bm{v}_2,\bm{u}_2}W_{\cU}(\bm{v}_1|\bm{u}_1)W_{\Lambda}(\bm{v}_2|\bm{u}_2)\tr[A_{\bm{u}_1}^TA_{\bm{u}_2}^T]\tr[A_{\bm{v}_1}A_{\bm{v}_2}] \\
&= \frac{1}{d^{4n}}\sum_{\bm{v}_1,\bm{u}_1,\bm{v}_2,\bm{u}_2}W_{\cU}(\bm{v}_1|\bm{u}_1)W_{\Lambda}(\bm{v}_2|\bm{u}_2)[d^{n}\delta_{\bm{u}_1,\bm{u}_2}][d^{n}\delta_{\bm{v}_1,\bm{v}_2}] \\
&= \frac{1}{d^{2n}}\sum_{\bm{u},\bm{v}}W_{\cU}(\bm{v}|\bm{u})W_{\Lambda}(\bm{v}|\bm{u}) \\
&= \sum_{\bm{u},\bm{v}} \frac{|W_{\cU}(\bm{v}|\bm{u})|}{\beta_{\cU}}\times [\frac{\operatorname{sgn}(W_{\cU}(\bm{v}|\bm{u}))W_{\Lambda}(\bm{v}|\bm{u})\beta_{\cU}}{d^{2n}}] \\
&\equiv \sum_{\bm{u},\bm{v}}{\rm Pr}(\bm{v},\bm{u}) X_{\bm{v},\bm{u}},
\end{align}
\end{subequations}

where $\bm{v},\bm{u} \in \bZ_d^n\times\bZ_d^n$, $\tr[\cU^{\dagger}\Lambda]$ is the Hilbert-Schmidt inner product between $\cU$ and $\Lambda$, $\operatorname{sgn}:\mathbb{R}\to\{1,-1\}$ is the sign function and $\beta_{\cN}:= \sum_{\bm{u},\bm{v}}|\cW_{\cN}(\bm{v}|\bm{u})|$.
${\rm Pr}(\bm{v},\bm{u}):= \frac{|W_{\cU}(\bm{v}|\bm{u})|}{\beta_{\cU}}$ is an valid probability distribution, and
\begin{align}
    X_{\bm{v},\bm{u}} :=\frac{\operatorname{sgn}(W_{\cU}(\bm{v}|\bm{u}))W_{\Lambda}(\bm{v}|\bm{u})\beta_{\cU}}{d^{2n}}
\end{align}
is a random variable. The estimation protocol is summarized in Protocol~\ref{protocol:fidelity-estimation-Chan-norm1}. We demonstrate that the DFE of unitary quantum channels via the $\ell_1$ norm is quantified by the mana of a quantum channel in Theorem~\ref{theo:channel_mana}:

\begin{theorem} [Informal]
\label{theo:channel_mana}
    Given a unitary channel $\cU \in  \cL(\cH_A^{\ox n}, \cH_{A'}^{\ox n})$, an additive error $\varepsilon$, and success probability $1- \delta$, we propose a fidelity estimation protocol via mana, the number of phase-space point operators sampled is  $K=\lceil 8\cdot2^{[\cM(\cU)]}/\varepsilon^2\delta)\rceil$,
    the sample complexity is $\cO\big(\frac{2^{[\cM(\cU)]}}{\varepsilon^2\delta} 
        + \frac{2^{[2\cdot\cM(\cU)]}}{\varepsilon^2}\ln\frac{1}{\delta}\big)$. 
\end{theorem}

\begin{proof}
    The proof remains in the performance analysis of protocol~\ref{protocol:fidelity-estimation-Chan-norm1} in the Appendix~\ref{appx:theo:channel_mana}.
\end{proof} 

\begin{algorithm}[!hbtp] 
\caption{Entanglement Fidelity estimation of channels via mana.} 
\label{protocol:fidelity-estimation-Chan-norm1}
\begin{algorithmic}[1]
\State{\textbf{Inputs:} \\
        \hskip1em\textbullet~~~The known target channel: $\cU$; \\
        \hskip1em\textbullet~~~The unknown prepared channel: $\Lambda$; \\
        \hskip1em\textbullet~~~The fixed additive error $\varepsilon$; \\
        \hskip1em\textbullet~~~The failure probability $\delta$.}
\State{\textbf{Output:} Estimation $\wt{Y}$, such that $\cF(\cU,\Lambda)$ lies 
        in $[\wt{Y}-\varepsilon, \wt{Y}+\varepsilon]$ with probability larger than $1-\delta$.}
\State{{Compute $K=\lceil 8\Delta_{\cU}/\varepsilon^2\delta)\rceil$.} 
       \Comment{{\color{gray}Number of phase-space point operators to be sampled}}}
\For{$k = 1, 2, \cdots, K$}\Comment{{\color{gray}Phase-space point operator sampling procedure}}     
\State{Randomly sample a pair of phase-space point operators $(A_{\bm{v}_k}, A_{\bm{u}_k})$ w.r.t. ${\rm Pr}(\bm{v}|\bm{u})=\frac{|W_{\cU}(\bm{v}|\bm{u})|}{\beta_{\cU}}$.}
\State{Compute the number of measurement rounds $N_k$ as 
\begin{align}
    {N_k = \left\lceil\frac{8\Delta_{\cU}}{K\varepsilon^2}\ln\frac{4}{\delta}\right\rceil.}
\end{align}
}
    \For{$j = 1, 2, \cdots, N_k$}\Comment{{\color{gray}Estimation procedure for $X_{\bm{u}}$}}
        \State{Prepare unitary channel $\cU$, apply it on the $A_{\bm{u}_k}$ and measure the outcome on the $A_{\bm{v}_k}$ basis.}
        \State{Record the eigenvalue $O_{j\vert k}\in[-1,1]$ of the corresponding outcome.}
    \EndFor
    \State{Compute the following empirical estimator from the experimental data $\{O_{j\vert k}\}_j$:
        \begin{align}
            \widetilde{X}(\bm{v}_k,\bm{u}_k) :=  \frac{1}{N_k}\sum_{j=1}^{N_k}\operatorname{sgn}(W_{\cU}(\bm{v}|\bm{u}))O_{j\vert k}\Delta_{\cU}.
            \qquad{\color{gray}\sim 
            X(\bm{v}_k,\bm{u}_k)= \operatorname{sgn}(W_{\cU}(\bm{v}|\bm{u}))W_{\Lambda}(\bm{v}_k|\bm{u}_k)\Delta_{\cU}
            }
        \end{align}
    }
\EndFor
\State{Compute the following empirical estimator from the data $\{\widetilde{X}(\bm{v}_k,\bm{u}_k)\}_k$:
\begin{align}
    \widetilde{Y} := \frac{1}{K}\sum_{k=1}^K\widetilde{X}(\bm{v}_k,\bm{u}_k).
            \qquad{\color{gray}\sim \frac{1}{d^{2n}}\tr[\cU^{\dagger}\Lambda]}
\end{align}
}
\State{Output $\widetilde{Y}$ as an unbiased estimator of $\cF(\cU,\Lambda) = \frac{\tr[\cU^{\dagger}\Lambda]}{d^{2n}}$.}
\end{algorithmic}
\end{algorithm}

\paragraph*{Summary of the performance of Protocol~\ref{protocol:fidelity-estimation-Chan-norm1}}

\begin{itemize}
  \item \textbf{Sample complexity:} 
     $\cO({\frac{\Delta_{\cU}}{\varepsilon^2\delta} + \Delta_{\cU}^2\times\frac{1}{\varepsilon^2}\ln\frac{1}{\delta}})$, where $\Delta_{\cU} = 2^{[\cM(\cU)]}$. 
  \item \textbf{Measurement complexity:} $\cO(d^{4n})$; All phase-space point operators in $\bZ_d^n\times\bZ_d^n$.
  \item \textbf{Measurement type:} single-qudit phase-space point measurements are enough.
\end{itemize}

Much like the well-conditioned states, for the Clifford operations, it holds that $\Delta_{\cU} =  1$, and thus the sample complexity deduces to $\cO({\frac{1}{\varepsilon^2\delta} + \frac{1}{\varepsilon^2}\ln\frac{1}{\delta}})$, which is consistent with the performance of Protocol~\ref{protocol:fidelity-estimation-Chan-norm2}.


\subsection{Well-conditioned channels}
If we concentrate on estimating the fidelity of Clifford unitaries, we can also achieve exponential speep up (compared
to the above estimation Protocol~\ref{protocol:fidelity-estimation-Chan-norm2} and Pretocol~\ref{protocol:fidelity-estimation-Chan-norm1}).

\paragraph{Protocol}

Our protocol relies crucially on a fact that $\chi(\cU) = |\cW_{\cU}| = d^{2n}$ for a Clifford unitary channel $\cU \in \cC_{d^{2n}}$.
Thus we have:
\begin{align}\label{eq:clifford-fidelity}
  \cF(\cU,\Lambda) 
&= \frac{1}{d^{2n}}\tr[\cU^{\dagger}\Lambda] \\
&= \frac{1}{d^{2n}}\sum_{\bm{u},\bm{v}}W_{\cU}(\bm{v}|\bm{u})W_{\Lambda}(\bm{v}|\bm{u}) \\
&= \sum_{\bm{u},\bm{v}}{\rm Pr}(\bm{v},\bm{u}) X_{\bm{v},\bm{u}} \label{eq:cl-f},
\end{align}
where ${\rm Pr}(\bm{v},\bm{u}) := \frac{W_{\cU}(\bm{v}|\bm{u})}{d^{2n}}$, and $X_{\bm{v},\bm{u}} := W_{\Lambda}(\bm{v}|\bm{u})$.
we know that $W_{\cU}(\bm{v}|\bm{u})$ is an valid probability distribution; 
Actually, it is a uniform distribution. The estimation protocol is summarized in Protocol~\ref{protocol:clifford-fidelity-estimation-DFE}, whose corresponding cost is introduced in Proposition~\ref{pro:clifford}.

\begin{proposition}[Informal]
\label{pro:clifford}
    Given a Clifford unitary channel $\cU \in  \cL(\cH_A^{\ox n}, \cH_{A'}^{\ox n})$, an additive error $\varepsilon$, and success probability $1- \delta$, we propose an efficient fidelity estimation protocol, the number of phase-space point operators sampled is  $K=\lceil 8/\varepsilon^2\delta)\rceil$,
    the sample complexity is $\cO(
       \frac{1}{\varepsilon^2}\ln\frac{1}{\delta})$. 
\end{proposition} 

\begin{proof}
    The proof remains in the performance analysis of protocol~\ref{protocol:clifford-fidelity-estimation-DFE} in the Appendix~\ref{appx:pro:clifford}.
\end{proof} 

\begin{algorithm}[!hbtp]
\caption{Fidelity estimation of clifford unitaries.}
\label{protocol:clifford-fidelity-estimation-DFE}
\begin{algorithmic}[1]
\State{\textbf{Inputs:} \\
        \hskip1em\textbullet~~~The known target Clifford unitary channel: $\cU$; \\
        \hskip1em\textbullet~~~The unknown prepared channel: $\Lambda$; \\
        \hskip1em\textbullet~~~The fixed additive error $\varepsilon$; \\
        \hskip1em\textbullet~~~The failure probability $\delta$.}
\State{\textbf{Output:} Estimation $\wt{Y}$, such that $\cF(\cU,\Lambda)$ lies 
        in $[\wt{Y}-\varepsilon, \wt{Y}+\varepsilon]$ with probability larger than $1-\delta$.}
\State{{Compute $K=\lceil 8/\varepsilon^2\ln(4/\delta)\rceil$.} 
       \Comment{{\color{gray}Number of phase-space point operators to be sampled}}}
\For{$k = 1, 2, \cdots, K$}\Comment{{\color{gray}Phase-space point operator sampling procedure}}     
\State{Randomly sample a pair of phase-space point operators $(A_{\bm{v}_k}, A_{\bm{u}_k})$ w.r.t. ${\rm Pr}(\bm{v}|\bm{u})=\frac{W_{\cU}(\bm{v}|\bm{u})}{d^{2n}}$.}
\State{Compute the number of measurement rounds $N_k$ as 
\begin{align}
    {N_k = \left\lceil\frac{8}{K\varepsilon^2}\ln\frac{4}{\delta}\right\rceil.}
\end{align}
}
    \For{$j = 1, 2, \cdots, N_k$}\Comment{{\color{gray}Estimation procedure for $X_{\bm{u}}$}}
        \State{Prepare unitary channel $\cU$, apply it on the $A_{\bm{u}_k}$ and measure the outcome on the $A_{\bm{v}_k}$ basis.}
        \State{Record the eigenvalue $O_{j\vert k}\in[-1,1]$ of the corresponding outcome.}
    \EndFor
    \State{Compute the following empirical estimator from the experimental data $\{O_{j\vert k}\}_j$:
        \begin{align}
            \widetilde{X}(\bm{v}_k,\bm{u}_k) :=  \frac{1}{N_k}\sum_{j=1}^{N_k}O_{j\vert k}.
            \qquad{\color{gray}\sim 
            X(\bm{v}_k,\bm{u}_k)= W_{\Lambda}(\bm{v}_k|\bm{u}_k)
            }
        \end{align}
    }
\EndFor
\State{Compute the following empirical estimator from the data $\{\widetilde{X}(\bm{v}_k,\bm{u}_k)\}_k$:
\begin{align}
    \widetilde{Y} := \frac{1}{K}\sum_{k=1}^K\widetilde{X}(\bm{v}_k,\bm{u}_k).
            \qquad{\color{gray}\sim \frac{1}{d^{2n}}\tr[\cU^{\dagger}\Lambda]}
\end{align}
}
\State{Output $\widetilde{Y}$ as an unbiased estimator of $\cF(\cU,\Lambda) = \frac{\tr[\cU^{\dagger}\Lambda]}{d^{2n}}$.}
\end{algorithmic}
\end{algorithm}

\subparagraph*{Summary of the performance of Protocol~\ref{protocol:clifford-fidelity-estimation-DFE}}

\begin{itemize}
  \item \textbf{Sample complexity:} $\cO( 1/\varepsilon^2\ln(1/\delta))$. 
  \item \textbf{Measurement complexity:} $\cO(d^{2n})$.
  \item \textbf{Measurement type:} single-qudit phase-space point measurements are enough.
\end{itemize}

\section{Concluding Remarks}
\label{sec:con}

In this work, we investigated the fundamental resource requirements for direct fidelity estimation in qudit systems with odd prime dimensions within the framework of the resource theory of nonstabilizerness. We demonstrated that direct fidelity estimation for general qudit states and unitaries requires exponential resources, as rigorously quantified by magic measures such as mana and Wigner rank. Importantly, we identified a class of states and operations—specifically those that are efficiently simulable classically—for which direct fidelity estimation remains practically feasible. Our results suggest that developing more efficient quantum certification protocols tailored to specific classes of useful states~\cite{yu2019optimal,zhu2019efficient,wang2019optimal}, rather than arbitrary states, may provide a more practical approach for quantum benchmarking.

The implementation of our protocols depends on the ability to efficiently measure phase-space point operators. Recent experimental advances in phase-space measurements for continuous-variable systems~\cite{waller2012phase,rundle2017simple} are encouraging, and suggest that similar progress in implementing discrete phase-space point operators~\cite{leonhardt1995quantum,khushwani2023selective,debrota2020discrete} may soon be achievable. Future work could explore similar phenomena and protocols in the continuous-variable quantum information setting. We expect these results to shed light on the role of non-stabilizer resources in quantum benchmarking and to guide the development of practical quantum benchmarking protocols in the FTQC regime.

\vspace{0.2in}
\textbf{Acknowledgements.}---Z.-P.L. and K.W. contributed equally to this work.  X.W. was partially supported by the National Key R\&D Program of China (Grant No.~2024YFB4504004),  the National Natural Science Foundation of China (Grant No.~12447107), the Guangdong Natural Science Foundation (Grant No.~2025A1515012834), the Guangdong Provincial Quantum Science Strategic Initiative (Grant No.~GDZX2403008, GDZX2403001), the Guangdong Provincial Key Lab of Integrated Communication, Sensing and Computation for Ubiquitous Internet of Things (Grant No.~2023B1212010007), the Quantum Science Center of Guangdong-Hong Kong-Macao Greater Bay Area, and the Education Bureau of Guangzhou Municipality.

\bibliographystyle{E:/Work/Wigner Rank/halpha}
\bibliography{E:/Work/Wigner Rank/references.bib}


\setcounter{secnumdepth}{2}
\appendix
\widetext

\section{Proof of Lemma~\ref{lemma:properties}}\label{appx:lemma:properties}

\begin{proof}[Proof of Lemma~\ref{lemma:properties}]

1. $A_\bu$ is Hermitian. Since $T_{\bm{u}}$ is unitary, it preserves Hermiticity, 
it suffices to show that $A_{\bm{0}}$ is Hermitian:
\begin{align}
    A_{\bm{0}}^\dagger
= \left(\frac{1}{d^n}\sum_{\bm{u}}T_{\bm{u}}\right)^\dagger
= \frac{1}{d^n}\sum_{\bm{u}}T_{\bm{u}}^\dagger
= \frac{1}{d^n}\sum_{\bm{u}}T_{-\bm{u}} 
= A_{\bm{0}}.
\end{align}

2. $\tr[A_\bu A_{\bu'}] = d^n\delta(\bu,\bu')$. 
First of all, we have
\begin{align}
\tr[A_{\bm{0}}^2]
= \tr[A_{\bm{0}}^\dagger A_{\bm{0}}]
= \frac{1}{d^{2n}}\sum_{\bm{u},\bm{v}}\tr[T^\dagger_{\bm{u}}T_{\bm{v}}] 
= \frac{1}{d^{2n}}\sum_{\bm{u}}\tr[T_{\bm{u}}^\dagger T_{\bm{u}}] 
= d^n.
\end{align}
On the other hand, it holds that
\begin{align}
 \tr[A_\bu A_{\bu'}]
&= \tr\left[T_\bu A_{\bm{0}} T_\bu^\dagger T_{\bu'} A_{\bm{0}} T_{\bu'}^\dagger \right] \\
&= \tr\left[(T_\bu^\dagger T_{\bu'})^\dagger A_{\bm{0}} (T_\bu^\dagger T_{\bu'}) A_{\bm{0}} \right] \\
&= \tr\left[T_{\bu'-\bu}^\dagger A_{\bm{0}} T_{\bu'-\bu} A_{\bm{0}} \right] \\
&= \begin{cases}
    \tr[A_{\bm{0}}^2], & \bu'=\bu \\
    0, & \bu'\neq\bu.
\end{cases}
\end{align}

3. $\tr[A_{\bm{u}}]=1$: Since $T_{\bm{u}}$ is unitary, it suffices to show that
\begin{align}
    1 = \tr[A_{\bm{0}}] = \frac{1}{d^n}\sum_{\bm{u}}\tr[T_{\bm{u}}] = \frac{1}{d^n}\tr[T_{\bm{0}}] .
\end{align}

4. We have
\begin{align}
  \norm{A_{\bu}}{\infty}
= \norm{T_{\bu} A_{\bm{0}} T_{\bu}}{\infty}
= \norm{A_{\bm{0}}}{\infty}
\leq \frac{1}{d^n}\sum_{\bm{u}}\norm{T_{\bu}}{\infty}
= 1,
\end{align}
where the second equality follows from that the operator norm is isometric invariant, and the inequality follows from the triangle inequality.

5. Follows since $\{A_{\bu}\}_{\bu}$ forms an orthogonal basis of the operator space.
Without loss of generality, we can express $X$ as
\begin{align}
    X = \sum_{\bm{u}} c_\bu A_\bu 
\end{align}
for some coefficients $c_\bu$. Under this decomposition, we have
\begin{align}
\tr[A_{\bm{v}}X] = \sum_{\bu}c_\bu\tr[A_{\bm{v}}A_\bu] = 
\begin{cases}
    d^nc_\bu, & \bm{v} = \bu \\
    0,        & \bm{v} \neq \bu.
\end{cases}
\end{align}
Thus, $c_u = \tr[A_{\bu}X]/d^n = W_X(\bu)$.

6. Follows directly from $4$.
\end{proof}

\section{Proof of Lemma~\ref{lemma:inner-product}}
\label{appx:lemma:inner-product}

\begin{proof}[Proof of Lemma~\ref{lemma:inner-product}]
By definition, we have
\begin{align}
    M = \sum_{\bm{u}\in\bZ_d^n\times\bZ_d^n}W_M(A_{\bm{u}})A_{\bm{u}}.
\end{align}
Then
\begin{align}
\tr[MN] 
&= \tr\left[\left(\sum_{\bm{u}}W_M(A_{\bm{u}})A_{\bm{u}}\right)
            \left(\sum_{\bm{v}}W_N(A_{\bm{v}})A_{\bm{v}}\right)\right] \\
&= \sum_{\bm{u},\bm{v}}W_M(A_{\bm{u}})W_N(A_{\bm{v}})\tr\left[A_{\bm{u}}A_{\bm{v}}\right] \\
&= d^n\sum_{\bm{u},\bm{v}}W_M(A_{\bm{u}})W_N(A_{\bm{v}})\delta_{\bm{u},\bm{v}} \\
&= d^n\sum_{\bm{u}}W_M(A_{\bm{u}})W_N(A_{\bm{u}}).
\end{align}
Specially, if $M=N=\proj{\psi}$ for some pure state $\ket{\psi}$, we have 
\begin{align}
    d^n\sum_{\bm{u}}W_M(A_{\bm{u}})W_N(A_{\bm{u}}) = \tr[\psi^2] = 1.
\end{align}
\end{proof}

\section{Proof of Theorem~\ref{theo:state_mana}}\label{appx:theo:state_mana}

The analysis of Protocol~\ref{protocol:fidelity-estimation-DFE-mana} also serves as the proof of Theorem~\ref{theo:state_mana}.
Protocol~\ref{protocol:fidelity-estimation-DFE-mana} is a two-stage estimation procedure: 
we first estimate each $\wt{X}_{\bm{u}_k}$ and then use them to estimate the target $\widetilde{Y}$:

\begin{align}
    O_{j\vert k} \quad\rightarrow\quad \wt{X}_{\bm{u}_k}
\quad\rightarrow\quad \wt{Y} 
\quad\xrightarrow{\text{\color{klevinblue}~Chebyshev inequality~}}\quad Y 
\quad\xrightarrow{\text{\color{klevinblue}~Hoeffding inequality~}}\quad \tr[\psi\rho].
\end{align}

\paragraph*{Step 1: Estimating $Y$.}
Assume that each $X_{\bm{u}}$ (equivalently, $W_\rho(\bm{u})$) can be estimated \textit{ideally}. 
Now say we want to estimate $\tr[\psi\rho]$ with some fixed additive error $\varepsilon/2$ 
and significance level. We sample $K$ times and obtain random variables $X_{\bm{u}_1},\cdots,X_{\bm{u}_K}$
and define the sample average $Y=1/K\sum_{k=1}^KX_{\bm{u}_k}$. The estimation error of $Y$ can be bounded using Chebyshev's inequality.

\textbf{Step 1.1: The expected value $\bE[Y]$ is finite.} 
This is trivial since
\begin{align}
  \bE[Y]  = \bE\left[\frac{1}{K}\sum_{k=1}^KX_{\bm{u}_k}\right] 
          = \sum_{\bm{u}\in\bZ_d^n\times\bZ_d^n}\frac{\vert W_\psi(\bm{u}) \vert}{\Delta_\psi}
            \operatorname{sgn}(W_\psi(\bm{u}))d^n\Delta_\psi W_\rho(\bm{u})
          = \tr[\psi\rho].
\end{align}
where $\bE$ denotes the expected value over the random choice of $k$.

\textbf{Step 1.2: The variance of $Y$ is bounded.} We compute the variance of $X_{\bm{u}_k}$:
\begin{align}
  \opn{Var}[X_{\bm{u}_k}]
&= \bE[X_{\bm{u}_k}^2] - (\bE[X_{\bm{u}_k}])^2 \\
&= \sum_{\bm{u}\in\bZ_d^n\times\bZ_d^n} \frac{\vert W_\psi(\bm{u}) \vert}{\Delta_\psi}
        (\operatorname{sgn}(W_\psi(\bm{u}))d^n\Delta_\psi W_\rho(\bm{u}))^2
    - \left(\tr[\psi\rho]\right)^2 \\
&= \sum_{\bm{u}\in\bZ_d^n\times\bZ_d^n} \vert W_\psi(\bm{u})\vert d^{2n} \Delta_\psi W^2_\rho(\bm{u}) 
                    - \tr[\psi\rho]^2 \\
&\leq \Delta_\psi\sum_{\bm{u}\in\bZ_d^n\times\bZ_d^n} d^nW_\rho^2(\bm{u}) - \tr[\psi\rho]^2 \\
&\leq \Delta_\psi\tr[\rho^2] \\
&\leq \Delta_\psi,
\end{align}
where the first inequality follows from (6) of Lemma~\ref{lemma:properties} and
the second inequality follows from Lemma~\ref{lemma:inner-product}.
Then, the variance of $Y$ follows directly
\begin{align}
    \opn{Var}[Y] = \frac{1}{K^2}\sum_{k=1}^K\opn{Var}[X_{\bm{u}_k}] \leq \frac{\Delta_\psi}{K}.
\end{align}

\textbf{Step 1.3: Chebyshev inequality.} We have
\begin{align}
\opn{Pr}\left\{\vert Y - \tr[\psi\rho]\vert\geq a\sqrt{\frac{\Delta_\psi}{K}}\right\}
\leq \opn{Pr}\left\{\vert Y - \tr[\psi\rho]\vert \geq a\sqrt{\opn{Var}[Y]}\right\}
\leq \frac{1}{a^2}.
\end{align}
Let $a=\sqrt{2/\delta}$ and $K=\lceil 8\Delta_\psi/\varepsilon^2\delta\rceil$ yields
\begin{align}\label{eq:estimation-bound-1}
    \opn{Pr}\left\{\vert Y - \tr[\psi\rho]\vert \geq \varepsilon/2\right\} \leq \delta/2. 
\end{align}

\paragraph*{Step 2: Estimating $X_{\bm{u}_k}$.}

We again employ an estimator $\wt{Y}$, which is obtained from a finite number of copies of the state $\rho$, to approximate the ideal infinite-precision estimator $Y$:
\begin{align}
    \wt{Y} 
&= \frac{1}{K}\sum_{k=1}^K\wt{X}_{\bm{u}_k} \\
&= \sum_{k=1}^K\sum_{j=1}^{N_k}\frac{1}{K}\times\frac{1}{N_k}\times
        \operatorname{sgn}(W_\psi(\bm{u}))d^n\Delta_\psi O_{j\vert k}  \\
&= \sum_{k=1}^K\sum_{j=1}^{N_k}\wt{O}_{j\vert k},
\end{align}
where
\begin{align}
  \wt{O}_{j\vert k} := \frac{\operatorname{sgn}(W_\psi(\bm{u}))d^n\Delta_\psi}{KN_k}O_{j\vert k} 
\in [-\Delta_\psi/(KN_k), \Delta_\psi/(KN_k)],
\end{align}
where the range bound follows since $O_{j\vert k}\in[-1/d^n,1/d^n]$ as proven in Lemma~\ref{lemma:properties}. 
That is to say, $\wt{Y}$ can be viewed as a sum of the (rescaled) random variables $\wt{O}_{j\vert k}$. 
Now we can apply another version of Hoeffding's inequality in terms of the sum, yielding:
\begin{subequations}\label{eq:estimation-bound-2}
\begin{align}
    \opn{Pr}\left(\vert\wt{Y} - Y \vert \geq \varepsilon/2\right)
&\leq 2\exp\left(- \frac{2\varepsilon^2}{4\sum_{k=1}^K\sum_{j=1}^{N_k}
                    \frac{4\Delta_\psi^2}{K^2N_k^2}}\right)\label{eq:1} \\
&= {2\exp\left(- \frac{\varepsilon^2}{\sum_{k=1}^K
                        \frac{8\Delta_\psi^2}{K^2N_k}}\right)}.
\end{align}
\end{subequations}
We need to cleverly choose $N_k$ so that the error probability satisfies
\begin{subequations}
\begin{align}
&\quad    {2\exp\left(- \frac{\varepsilon^2}
            {\sum_{k=1}^K\frac{8\Delta_\psi^2}{K^2N_k}}\right)} = \delta/2 \\
\Leftarrow&\quad \sum_{k=1}^K\frac{8\Delta_\psi^2}{K^2N_k} 
                = \frac{\varepsilon^2}{\ln\frac{4}{\delta}} \\
\Leftarrow&\quad  \forall k=1,\cdots, K,\quad 
                  \frac{8\Delta_\psi^2}{K^2N_k} 
                = \frac{1}{K}\times\frac{\varepsilon^2}{\ln\frac{4}{\delta}} \\
\Leftarrow&\quad  \forall k=1,\cdots, K,\quad 
    N_k = \left\lceil\frac{8\Delta_\psi^2}{K\varepsilon^2}\ln\frac{4}{\delta}\right\rceil.
\end{align}
\end{subequations}

\paragraph*{Step 3: Union bound.}

From the above two steps, we can derive the union bound:
\begin{equation}
    \Pr\left\{\vert \wt{Y} - \tr[\psi\rho]\vert \leq \varepsilon\right\} \geq 1 - \delta,
\end{equation}
and conclude that,
\begin{quote}
\textit{\color{googlered}
With probability $\geq 1-\delta$, 
the fidelity $\cF(\psi,\rho)$ lies in the range $[\wt{Y}-\varepsilon, \wt{Y}+\varepsilon]$.}
\end{quote}

\paragraph*{Number of copies}

This DFE protocol via mana first samples $K=\lceil 8\Delta_\psi/\varepsilon^2\delta)\rceil$
number of phase-space point operators and uses 
$N_k=\lceil8\Delta_\psi^2/(K\varepsilon^2)\ln(4/\delta)\rceil$
number of state copies to estimate the corresponding expectation value.
The expected value of $N_k$ w.r.t. $k$ is given by
\begin{align}
    \bE[N_k] 
&= \sum_{\bm{u}\in\bZ_d^n\times\bZ_d^n} {\rm Pr}(\bm{u})N_{\bm{u}} \\
&\leq 1 + \sum_{\bm{u}\in\bZ_d^n\times\bZ_d^n}\frac{\vert W_\psi(\bm{u}) \vert}{\Delta_\psi}\times 
            \frac{8\Delta_\psi^2}{K\varepsilon^2}\ln\frac{4}{\delta} \\
&= 1 + \frac{8\Delta_\psi^2}{K\varepsilon^2}\ln\frac{4}{\delta}.
\end{align}
Thus, the total number of samples consumed is given by
\begin{align}
    \bE[N] = \sum_{k=1}^K \bE[N_k]
\leq 1  + \frac{8\Delta_\psi}{\varepsilon^2\delta} 
        + 8\Delta_\psi^2\times\frac{1}{\varepsilon^2}\ln\frac{4}{\delta},
\end{align}
and the sample complexity is $\cO(\frac{2^{\cM(\psi)}}{\varepsilon^2\delta} 
        + \frac{2^{2\cdot\cM(\psi)}}{\varepsilon^2}\ln\frac{1}{\delta})$ with $\Delta_\psi = 2^{\cM(\psi)}$.


\section{Proof of Proposition~\ref{pro:stab}}\label{appx:pro:stab}

The analysis of Protocol~\ref{protocol:stabilizer-fidelity-estimation-DFE} serves as the proof of Proposition~\ref{pro:stab}. Protocol~\ref{protocol:stabilizer-fidelity-estimation-DFE} is also a two-stage estimation procedure: 
we first estimate each $\widetilde{W}_\rho(\bm{u}_k)$ and then use them to estimate the target $\widetilde{Y}$:

\begin{align}
    O_{j\vert k} \quad\rightarrow\quad \widetilde{W}_\rho(\bm{u}_k) \quad\rightarrow\quad \wt{Y} 
\quad\xrightarrow{\text{\color{klevinblue}~Hoeffding inequality~}}\quad Y 
\quad\xrightarrow{\text{\color{klevinblue}~Hoeffding inequality~}}\quad \tr[\psi\rho].
\end{align}

\subparagraph*{Step 1: Estimating $Y$.}
Assume that each $\wh{W}_\rho(\bm{u})$ can be estimated \textit{ideally}. 
We estimate $\tr[\psi\rho]$ with some fixed additive error $\varepsilon/2$ 
and significance level (aka. failure probability) $\delta/2$. 
We sample $K$ times and obtain random variables $\wh{W}_\rho(\bm{u}_1),\cdots,\wh{W}_\rho(\bm{u}_K)$
and define the sample average $Y=1/K\sum_{k=1}^K\wh{W}_\rho(\bm{u}_k)$.
Since $\wh{W}_\rho(\bm{u})=\tr[A_{\bm{u}}\rho]\in[-1,1]$ which is proved in Lemma~\ref{lemma:properties}, 
we can apply the Hoeffding inequality to obtain a tighter bound on the estimation error of $Y$.

First of all, we have
\begin{align}
  \bE[Y]  = \bE\left[\frac{1}{K}\sum_{k=1}^K\wh{W}_\rho(\bm{u}_k)\right] 
          = d^n\sum_{\bm{u}\in\bZ_d^n\times\bZ_d^n}W_\psi(\bm{u})W_\rho(\bm{u}) = \tr[\psi\rho].
\end{align}
where $\bE$ denotes the expected value over the random choice of $k$.
Applying the above Hoeffding inequality to $Y$ we obtain
\begin{align}
  \Pr\left\{\left\vert Y - \bE[Y]\right\vert \geq \varepsilon/2\right\}
= \Pr\left\{\left\vert Y - \tr[\psi\rho]\right\vert \geq \varepsilon/2\right\}
\leq 2\exp\left( -\frac{K^2\varepsilon^2}{2\sum_{k=1}^K2^2} \right)
= 2\exp\left( -\frac{K\varepsilon^2}{8} \right).
\end{align}

To ensure that the failure probability is upper bounded by $\delta/2$, we set the number of phase-space operators to be:
\begin{equation}
    K=\lceil8/\varepsilon^2\ln(4/\delta)\rceil.
\end{equation}

There is a key distinction between
Protocol~\ref{protocol:fidelity-estimation-DFE-mana} and Protocol~\ref{protocol:stabilizer-fidelity-estimation-DFE}:
In Protocol~\ref{protocol:fidelity-estimation-DFE-mana} we employ a weaker Chebyshev inequality to bound the number of phase-space operators $K$, 
whereas in Protocol~\ref{protocol:stabilizer-fidelity-estimation-DFE}, we apply the Hoeffding inequality. The latter results in an exponentially improved dependence on the significance level $\delta$.

\subparagraph*{Step 2: Estimating $\wh{W}_\rho(\bm{u}_k)$.}
Rather than directly estimating each individual $\wh{W}_\rho(\bm{u}_k)$, we approximate the ideal infinite-precision estimator $Y$ by a finite-precision $\wt{Y}$, which is derived from a finite number of copies of the state $\rho$:
\begin{align}
    \wt{Y} 
&= \frac{1}{K}\sum_{k=1}^K\wt{W}_\rho(\bm{u}_k) \\
&= \sum_{k=1}^K\sum_{j=1}^{N_k} \frac{1}{K}\times\frac{1}{N_k}\times O_{j\vert k} \\
&= \sum_{k=1}^K\sum_{j=1}^{N_k}\wt{O}_{jk},
\end{align}
where
\begin{align}
  \wt{O}_{j\vert k} := \frac{1}{KN_k}O_{j\vert k} \in [-1/(KN_k), 1/(KN_k)].
\end{align}
The range bound follows since $O_{j\vert k}\in\{-1,1\}$, and $\wt{Y}$ can be viewed as a sum of random variables $\wt{O}_{j\vert k}$. 
Now we can apply Hoeffding's inequality again, yielding:
\begin{subequations}\label{eq:estimation-bound-stab-2}
\begin{align}
    \opn{Pr}\left(\left\vert\wt{Y} - Y \right\vert \geq \varepsilon/2\right)
&\leq 2\exp\left(- \frac{2\varepsilon^2}{4\sum_{k=1}^K\sum_{j=1}^{N_k}(u_b - l_b)^2}\right) \\
&= 2\exp\left(- \frac{\varepsilon^2}{2\sum_{k=1}^K\sum_{j=1}^{N_k}\frac{4}{K^2N_k^2}}\right) \\
&= {2\exp\left(- \frac{\varepsilon^2}{\sum_{k=1}^K\frac{8}{K^2N_k}}\right)}.
\end{align}
\end{subequations}
We need to cleverly choose $N_k$ so that the error probability satisfies
\begin{subequations}
\begin{align}
&\quad    {2\exp\left(- \frac{\varepsilon^2}{\sum_{k=1}^K\frac{8}{K^2N_k}}\right)} = \delta/2 \\
\Leftarrow&\quad \sum_{k=1}^K\frac{8}{K^2N_k} = \frac{\varepsilon^2}{\ln\frac{4}{\delta}} \\
\Leftarrow&\quad  \forall k=1,\cdots, K,\quad 
                  \frac{8}{K^2N_k} = \frac{1}{K}\times\frac{\varepsilon^2}{\ln\frac{4}{\delta}} \\
\Leftarrow&\quad  \forall k=1,\cdots, K,\quad 
                  N_k = \left\lceil\frac{8}{K\varepsilon^2}\ln\frac{4}{\delta}\right\rceil = 1.
\end{align}
\end{subequations}

\subparagraph*{Step 3: Union bound}

From the previous two steps, we have the following lower bound on the probability: 
\begin{align}
  \Pr\left\{\left\vert \wt{Y} - \bE[Y]\right\vert \leq \varepsilon\right\}
&= 1 - \Pr\left\{\left\vert \wt{Y} - \bE[Y]\right\vert \geq \varepsilon\right\}  \\
&\geq 1 - \left( \Pr\left\{\left\vert Y - \bE[Y]\right\vert \geq \varepsilon/2\right\}
               + \opn{Pr}\left(\vert\wt{Y} - Y \vert \geq \varepsilon/2\right)\right) \\
&\geq 1 - \delta.
\end{align}
We can then conclude that,
\begin{quote}
\textit{\color{googlered}
With probability $\geq 1-\delta$, 
the fidelity $\cF(\psi,\rho)$ lies in the range $[\wt{Y}-\varepsilon, \wt{Y}+\varepsilon]$.}
\end{quote}

\subparagraph*{Number of copies}

Note that the DFE method samples $K=\lceil 8/\varepsilon^2\ln(4/\delta)\rceil$ 
Pauli observables, with the number of samples independent of the system size. For each sampled $k$, a constant number of copies $N_k = 1$ is required. Therefore, the total number of samples is given by:
\begin{align}
    \bE[N] = \sum_{k=1}^K N_k
      = \sum_{k=1}^K \left\lceil\frac{8}{K\varepsilon^2}\ln\frac{4}{\delta}\right\rceil
      \sim \left\lceil\frac{8}{\varepsilon^2}\ln\frac{4}{\delta}\right\rceil.
\end{align}
Then the sample complexity is $\cO(\frac{1}{\varepsilon^2}\ln\frac{1}{\delta})$.


\section{Proof of Theorem~\ref{theo:channel_Wigner_rank}}\label{appx:theo:channel_Wigner_rank}

The analysis of Protocol~\ref{protocol:fidelity-estimation-Chan-norm2} serves as the proof of Theorem~\ref{theo:channel_Wigner_rank}. Protocol~\ref{protocol:fidelity-estimation-Chan-norm2} is a two-stage estimation procedure:
we first estimate each $\wt{X}_{\bm{v}_k,\bm{u}_k}$ and then use them to estimate the target $\widetilde{Y}$:

\begin{align}
    O_{j\vert k} \quad\rightarrow\quad \wt{X}_{\bm{v}_k,\bm{u}_k}
\quad\rightarrow\quad \wt{Y} 
\quad\xrightarrow{\text{\color{klevinblue}~Chebyshev inequality~}}\quad Y 
\quad\xrightarrow{\text{\color{klevinblue}~Hoeffding inequality~}}\quad \frac{\tr[\cU^{\dagger}\Lambda]}{d^{2n}}.
\end{align}

\paragraph*{Step 1: Estimating $Y$.}
Assume that each $X_{\bm{v},\bm{u}}$ (equivalently, $W_{\Lambda}(\bm{v}|\bm{u})$) can be estimated \textit{ideally}. 
Now say we want to estimate $\frac{\tr[\cU^{\dagger}\Lambda] }{d^{2n}}$ with some fixed additive error $\varepsilon/2$ 
and significance level (aka. failure probability) $\delta/2$. 
We sample $K$ times and obtain random variables $X_{(\bm{v}_1,\bm{u}_1)},\cdots,X_{(\bm{v}_K,\bm{u}_K)}$
and define the sample average $Y=1/K\sum_{k=1}^KX_{(\bm{v}_k,\bm{u}_k)}$.
We can apply the Chebyshev inequality to bound the estimation error of $Y$.

\textbf{Step 1.1: The expected value $\bE[Y]$ is finite.} 
This is trivial.

\textbf{Step 1.2: The variance of $Y$ is bounded.} Since $Y$ is a weighted sum of $X_{(\bm{v}_k,\bm{u}_k)}$, 
we first compute the variance of $X_{(\bm{v}_k,\bm{u}_k)}$:
\begin{align}
  \opn{Var}[X_{(\bm{v}_k,\bm{u}_k)}]
&= \bE[X_{(\bm{v}_k,\bm{u}_k)}^2] - (\bE[X_{(\bm{v}_k,\bm{u}_k)}])^2 \\
&= \sum_{\bm{u},\bm{v}} \frac{W^2_{\cU}(\bm{v}|\bm{u}) }{d^{2n}}\times \frac{W^2_{\Lambda}(\bm{v}|\bm{u})}{W^2_{\cU}(\bm{v}|\bm{u})}
    -\frac{\tr^2[\cU^{\dagger}\Lambda]}{d^{4n}}\\
&=  \sum_{\bm{u},\bm{v}} \frac{W^2_{\Lambda}(\bm{v}|\bm{u}) }{d^{2n}}
    -\frac{\tr^2[\cU^{\dagger}\Lambda]}{d^{4n}} \\
&= \frac{1}{d^{2n}}\tr[\Lambda^{\dagger}\Lambda] - \frac{\tr^2[\cU^{\dagger}\Lambda]}{d^{4n}}\\
&\leq \frac{1}{d^{2n}}\tr[\Lambda^{\dagger}\Lambda] \\
&\leq 1, 
\end{align}
Then, the variance of $Y$ follows directly 
\begin{align}
    \opn{Var}[Y] = \frac{1}{K^2}\sum_{k=1}^K\opn{Var}[X_{(\bm{v}_k,\bm{u}_k)}] \leq \frac{1}{K}.
\end{align}

\textbf{Step 1.3: Chebyshev inequality.} We have
\begin{align}
\opn{Pr}\left\{\vert Y - \frac{1}{d^{2n}}\tr[\cU ^{\dagger}\Lambda]\vert \geq\frac{a}{\sqrt{K}}\right\}
\leq \opn{Pr}\left\{\vert Y - \frac{1}{d^{2n}}\tr[\cU ^{\dagger}\Lambda]\vert \geq a \sqrt{\opn{Var}[Y]}\right\} \leq \frac{1}{a^2}.
\end{align}
Let $a=\sqrt{2/\delta}$ and \textbf{$K=\lceil 8/\varepsilon^2\delta)\rceil$} yields
\begin{align}\label{eq:wigerrank-channel-estimation-bound-1}
    \opn{Pr}\left\{\vert Y -\frac{1}{d^{2n}}\tr[\cU ^{\dagger}\Lambda]\vert \geq \varepsilon/2\right\} \leq \delta/2.
\end{align}

\paragraph*{Step 2: Estimating $X_{(\bm{v}_k,\bm{u}_k)}$.}

\begin{align}
    \wt{Y} 
&= \frac{1}{K}\sum_{k=1}^K\wt{X}_{(\bm{v}_k,\bm{u}_k)} \\
&= \sum_{k=1}^K\sum_{j=1}^{N_k}\frac{1}{K}\times\frac{1}{N_k}\times\frac{O_{j\vert k}}{W_{\cU}(\bm{v}_k,\bm{u}_k)} \\
&= \sum_{k=1}^K\sum_{j=1}^{N_k}\wt{O}_{j\vert k},
\end{align}
where
\begin{align}
  \wt{O}_{j\vert k} := \frac{1}{KN_kW_{\cU}(\bm{v}_k,\bm{u}_k)}O_{j\vert k} 
\in [-\frac{\Delta}{KN_kW_{\cU}(\bm{v}_k,\bm{u}_k)}, \frac{\Delta}{KN_kW_{\cU}(\bm{v}_k,\bm{u}_k)}],
\end{align}
where the range bound follows since $O_{j\vert k}\in[-\Delta,\Delta]$, 
where $\Delta=1$. 
Now we apply the Hoeffding's inequality, yielding:
\begin{subequations}\label{eq:wigerrank-channel-estimation-bound-2}
\begin{align}
    \opn{Pr}\left(\vert\wt{Y} - Y \vert \geq \varepsilon/2\right)
&\leq 2\exp\left(- \frac{2\varepsilon^2}{4\sum_{k=1}^K\sum_{j=1}^{N_k}
                    \frac{4\Delta^2}{K^2N_k^2W^2_{\cU}(\bm{v}_k,\bm{u}_k)}}\right) \\
&= {2\exp\left(- \frac{\varepsilon^2}{\sum_{k=1}^K
                        \frac{8\Delta^2}{K^2N_kW^2_{\cU}(\bm{v}_k,\bm{u}_k)}}\right)}.
\end{align}
\end{subequations}
We need to cleverly choose $N_k$ so that the error probability satisfies
\begin{subequations}
\begin{align}
&\quad    {2\exp\left(- \frac{\varepsilon^2}
            {\sum_{k=1}^K\frac{8\Delta}{K^2N_kW^2_{\cU}(\bm{v}_k,\bm{u}_k)}}\right)} = \delta/2 \\
\Leftarrow&\quad \sum_{k=1}^K\frac{8\Delta^2}{K^2N_kW^2_{\cU}(\bm{v}_k,\bm{u}_k)} 
                = \frac{\varepsilon^2}{\ln\frac{4}{\delta}} \\
\Leftarrow&\quad  \forall k=1,\cdots, K,\quad 
                  \frac{8\Delta^2}{K^2N_kW^2_{\cU}(\bm{v}_k,\bm{u}_k)} 
                = \frac{1}{K}\times\frac{\varepsilon^2}{\ln\frac{4}{\delta}} \\
\Leftarrow&\quad  \forall k=1,\cdots, K,\quad 
    N_k = \left\lceil\frac{8\Delta^2}{KW^2_{\cU}(\bm{v}_k,\bm{u}_k)\varepsilon^2}\ln\frac{4}{\delta}\right\rceil.
\end{align}
\end{subequations}

\paragraph*{Step 3: Union bound.}

From the above two steps, we obtain the following two probability bounds
given in Eqs.~\eqref{eq:wigerrank-channel-estimation-bound-1} and~\eqref{eq:wigerrank-channel-estimation-bound-2}:
\begin{align}
  \Pr\left\{\vert Y - \frac{1}{d^{2n}}\tr[\cU ^{\dagger}\Lambda]\vert \geq \varepsilon/2\right\}
&\leq \delta/2, \\
\opn{Pr}\left(\vert\wt{Y} - Y \vert \geq \varepsilon/2\right)
&\leq \delta/2.
\end{align}
Using the union bound, we have 
\begin{align}
  \Pr\left\{\vert \wt{Y} - \frac{1}{d^{2n}}\tr[\cU ^{\dagger}\Lambda]\vert \leq \varepsilon\right\}
&= 1 - \Pr\left\{\vert \wt{Y} - \bE[Y]\vert \geq \varepsilon\right\}  \\
&\geq 1 - \left( \Pr\left\{\vert Y - \bE[Y]\vert \geq \varepsilon/2\right\}
               + \opn{Pr}\left(\vert\wt{Y} - Y \vert \geq \varepsilon/2\right)\right) \\
&\geq 1 - \delta.
\end{align}
We can then conclude that,
\begin{quote}
\textit{\color{googlered}
With probability $\geq 1-\delta$, 
the fidelity $\cF(\cU,\Lambda)$ lies in the range $[\wt{Y}-\varepsilon, \wt{Y}+\varepsilon]$.}
\end{quote}

\paragraph*{Number of copies}

Notice that the DFE method first samples \textbf{$K=\lceil 8/\varepsilon^2\delta)\rceil$}
number of phase-space point operators and use 
\textbf{$N_k=\lceil8\Delta^2/(KW^2_{\cU}(\bm{v}_k,\bm{u}_k)\varepsilon^2)\ln(4/\delta)\rceil$}
number of unitary copies to estimate the corresponding expectation value.
The expected value of $N_k$ w.r.t. $k$ is given by
\begin{align}
    \bE[N_k] 
&= \sum_{\bm{v},\bm{u}} {\rm Pr}(\bm{v},\bm{u})N_{\bm{v},\bm{u}} \\
&\leq \sum_{\bm{v},\bm{u}} {\rm Pr}(\bm{v},\bm{u})[1+ \frac{8\Delta^2}{(KW^2_{\cU}(\bm{v}_k,\bm{u}_k)\varepsilon^2}\ln\frac{4}{\delta}] \\
&\leq 1 + \sum_{\bm{v},\bm{u}} \frac{W^2_{\cU} (\bm{v}|\bm{u}) }{d^{2n}} \times 
        \frac{8\Delta^2}{KW^2_{\cU} (\bm{v}|\bm{u})\varepsilon^2}\ln\frac{4}{\delta} \\
&\leq 1 + \frac{8\chi(\cU)}{K\varepsilon^2d^{2n}}\ln\frac{4}{\delta},
\end{align}
where the last inequality follows from the fact that $\Delta=1$ and the definition of $\chi(\cU)$. 
Thus, the total number of samples consumed is given by
\begin{align}
    \bE[N] = \sum_{k=1}^K \bE[N_k]
\leq  1+ \frac{8}{\varepsilon^2\delta} + 
{8\cdot2^{\chi_{\log}(\cU)}}\times\frac{1}{\varepsilon^2}\ln\frac{4}{\delta},
\end{align}
and the sample complexity is $\cO( \frac{1}{\varepsilon^2\delta} + 
{2^{\chi_{\log}(\cU)}}\times\frac{1}{\varepsilon^2}\ln\frac{1}{\delta} )$.

\section{Proof of Theorem~\ref{theo:channel_mana}}\label{appx:theo:channel_mana}

The performance analysis of Protocol~\ref{protocol:fidelity-estimation-Chan-norm1} is the proof of Theorem~\ref{theo:channel_mana}. Protocol~\ref{protocol:fidelity-estimation-Chan-norm1} is a two-stage estimation procedure: we first estimate each $\wt{X}_{\bm{v}_k,\bm{u}_k}$ and then use them to estimate the target $\widetilde{Y}$:

\begin{align}
    O_{j\vert k} \quad\rightarrow\quad \wt{X}_{\bm{v}_k,\bm{u}_k}
\quad\rightarrow\quad \wt{Y} 
\quad\xrightarrow{\text{\color{klevinblue}~Chebyshev inequality~}}\quad Y 
\quad\xrightarrow{\text{\color{klevinblue}~Hoeffding inequality~}}\quad \frac{\tr[\cU^{\dagger}\Lambda]}{d^{2n}}.
\end{align}

\paragraph*{Step 1: Estimating $Y$.}
Assume that each $X_{\bm{v},\bm{u}}$ (equivalently, $W_{\Lambda}(\bm{v}|\bm{u})$) can be estimated \textit{ideally}. 
Now say we want to estimate $\frac{\tr[\cU^{\dagger}\Lambda] }{d^{2n}}$ with some fixed additive error $\varepsilon/2$ 
and significance level (aka. failure probability) $\delta/2$. 
We sample $K$ times and obtain random variables $X_{(\bm{v}_1,\bm{u}_1)},\cdots,X_{(\bm{v}_K,\bm{u}_K)}$
and define the sample average $Y=1/K\sum_{k=1}^KX_{(\bm{v}_k,\bm{u}_k)}$.
We can apply the Chebyshev inequality to bound the estimation error of $Y$.

\textbf{Step 1.1: The expected value $\bE[Y]$ is finite.} 
This is trivial.

\textbf{Step 1.2: The variance of $Y$ is bounded.} Since $Y$ is a weighted sum of $X_{(\bm{v}_k,\bm{u}_k)}$, 
we first compute the variance of $X_{(\bm{v}_k,\bm{u}_k)}$:
\begin{align}
  \opn{Var}[X_{(\bm{v}_k,\bm{u}_k)}]
&= \bE[X_{(\bm{v}_k,\bm{u}_k)}^2] - (\bE[X_{(\bm{v}_k,\bm{u}_k)}])^2 \\
&= \sum_{\bm{u},\bm{v}} \frac{|W_{\cU}(\bm{v}|\bm{u})|}{\beta_{\cU}}\times \frac{W^2_{\Lambda}(\bm{v}|\bm{u})\beta^2_{\cU}}{d^{4n}}
    -\frac{\tr^2[\cU^{\dagger}\Lambda]}{d^{4n}}\\
&= \sum_{\bm{u},\bm{v}} \frac{|W_{\cU}(\bm{v}|\bm{u})|}{d^{4n}} W^2_{\Lambda}(\bm{v}|\bm{u})\beta_{\cU}
    -\frac{\tr^2[\cU^{\dagger}\Lambda]}{d^{4n}}\\
& \leq \sum_{\bm{u},\bm{v}} \frac{W^2_{\Lambda}(\bm{v}|\bm{u})\beta_{\cU}}{d^{4n}} 
    -\frac{\tr^2[\cU^{\dagger}\Lambda]}{d^{4n}}\\
& \leq \frac{\beta_{\cU}}{d^{2n}}\tr[\Lambda^{\dagger}\Lambda] - \frac{\tr^2[\cU^{\dagger}\Lambda]}{d^{4n}}\\
&\leq \frac{\beta_{\cU}}{d^{4n}}\tr[\Lambda^{\dagger}\Lambda] \\
&\leq \frac{\beta_{\cU}}{d^{2n}} \\
&\leq 2^{[{\cM(\cU)}]},
\end{align}

Then, the variance of $Y$ follows directly
\begin{align}
    \opn{Var}[Y] = \frac{1}{K^2}\sum_{k=1}^K\opn{Var}[X_{(\bm{v}_k,\bm{u}_k)}] \leq \frac{2^{[{\cM(\cU)}]}}{K}.
\end{align}

Here we define $\Delta_{\cN} = 2^{[{\cM(\cN)}]}$ for a quantum channel $\cN$, the exponential version of the mana of channel $\cN$.

\textbf{Step 1.3: Chebyshev inequality.} We have
\begin{align}
  \opn{Pr}\left\{\vert Y - \frac{1}{d^{2n}}\tr[\cU ^{\dagger}\Lambda]\vert \geq a \sqrt{\frac{\Delta_{\cU}}{K}}\right\}
\leq \opn{Pr}\left\{\vert Y - \frac{1}{d^{2n}}\tr[\cU ^{\dagger}\Lambda]\vert \geq a \sqrt{\opn{Var}[Y]}\right\} \leq \frac{1}{a^2}.
\end{align}
Let $a=\sqrt{2/\delta}$ and \textbf{$K=\lceil 8\Delta_{\cU}/\varepsilon^2\delta)\rceil$} yields
\begin{align}\label{eq:estimation-bound-channel_mana-1}
    \opn{Pr}\left\{\vert Y -\frac{1}{d^{2n}}\tr[\cU ^{\dagger}\Lambda]\vert \geq \varepsilon/2\right\} \leq \delta/2.
\end{align}

\paragraph*{Step 2: Estimating $X_{(\bm{v}_k,\bm{u}_k)}$.}

\begin{align}
    \wt{Y} 
&= \frac{1}{K}\sum_{k=1}^K\wt{X}_{(\bm{v}_k,\bm{u}_k)} \\
&= \sum_{k=1}^K\sum_{j=1}^{N_k}\frac{1}{K}\times\frac{1}{N_k}\times \frac{\operatorname{sgn}(W_{\cU}(\bm{v}|\bm{u}))O_{j\vert k}\beta_{\cU}}{d^{2n}}\\
&= \sum_{k=1}^K\sum_{j=1}^{N_k}\wt{O}_{j\vert k},
\end{align}
where
\begin{align}
  \wt{O}_{j\vert k} := \frac{\operatorname{sgn}(W_{\cU}(\bm{v}|\bm{u}))\beta_{\cU}}{d^{2n}KN_k}O_{j\vert k} 
\in [-\frac{\Delta_{\cU}}{KN_k}, \frac{\Delta_{\cU}}{KN_k}],
\end{align}
where the range bound follows since $O_{j\vert k}\in[-1, 1]$ and $\frac{\beta_{\cU}}{d^{2n}} \leq \Delta_{\cU} $.
Now we apply the Hoeffding's inequality, yielding:
\begin{subequations}\label{eq:estimation-bound-channel_mana-2}
\begin{align}
    \opn{Pr}\left(\vert\wt{Y} - Y \vert \geq \varepsilon/2\right)
&\leq 2\exp\left(- \frac{2\varepsilon^2}{4\sum_{k=1}^K\sum_{j=1}^{N_k}
                    \frac{4\Delta_{\cU}^2}{K^2N_k^2}}\right) \\
&= {2\exp\left(- \frac{\varepsilon^2}{\sum_{k=1}^K
                        \frac{8\Delta_{\cU}^2}{K^2N_k}}\right)}.
\end{align}
\end{subequations}
We need to cleverly choose $N_k$ so that the error probability satisfies
\begin{subequations}
\begin{align}
&\quad    {2\exp\left(- \frac{\varepsilon^2}{\sum_{k=1}^K
                        \frac{8\Delta_{\cU}^2}{K^2N_k}}\right)} = \delta/2 \\
\Leftarrow&\quad \sum_{k=1}^K\frac{8\Delta_{\cU}^2}{K^2N_k} 
                = \frac{\varepsilon^2}{\ln\frac{4}{\delta}} \\
\Leftarrow&\quad  \forall k=1,\cdots, K,\quad 
                  \frac{8\Delta_{\cU}^2}{K^2N_k} 
                = \frac{1}{K}\times\frac{\varepsilon^2}{\ln\frac{4}{\delta}} \\
\Leftarrow&\quad  \forall k=1,\cdots, K,\quad 
    N_k = \left\lceil\frac{8\Delta_{\cU}^2}{K\varepsilon^2}\ln\frac{4}{\delta}\right\rceil.
\end{align}
\end{subequations}

\paragraph*{Step 3: Union bound.}

From the above two steps, we obtain the following two probability bounds
given in Eqs.~\eqref{eq:estimation-bound-channel_mana-1} and~\eqref{eq:estimation-bound-channel_mana-2}:
\begin{align}
  \Pr\left\{\vert Y - \frac{1}{d^{2n}}\tr[\cU ^{\dagger}\Lambda]\vert \geq \varepsilon/2\right\}
&\leq \delta/2, \\
\opn{Pr}\left(\vert\wt{Y} - Y \vert \geq \varepsilon/2\right)
&\leq \delta/2.
\end{align}
Using the union bound, we have 
\begin{align}
  \Pr\left\{\vert \wt{Y} - \frac{1}{d^{2n}}\tr[\cU ^{\dagger}\Lambda]\vert \leq \varepsilon\right\}
&= 1 - \Pr\left\{\vert \wt{Y} - \bE[Y]\vert \geq \varepsilon\right\}  \\
&\geq 1 - \left( \Pr\left\{\vert Y - \bE[Y]\vert \geq \varepsilon/2\right\}
               + \opn{Pr}\left(\vert\wt{Y} - Y \vert \geq \varepsilon/2\right)\right) \\
&\geq 1 - \delta.
\end{align}
We can then conclude that,
\begin{quote}
\textit{\color{googlered}
With probability $\geq 1-\delta$, 
the fidelity $\cF(\cU,\Lambda)$ lies in the range $[\wt{Y}-\varepsilon, \wt{Y}+\varepsilon]$.}
\end{quote}

\paragraph*{Number of copies}

Notice that the DFE method first samples \textbf{$K=\lceil \frac{8\Delta_{\cU}}{\varepsilon^2\delta}\rceil$}
number of phase-space point operators and use 
\textbf{$N_k = \left\lceil\frac{8\Delta_{\cU}^2}{K\varepsilon^2}\ln\frac{4}{\delta}\right\rceil$}
number of unitary copies to estimate the corresponding expectation value.
The expected value of $N_k$ w.r.t. $k$ is given by
\begin{align}
    \bE[N_k] 
&= \sum_{\bm{v},\bm{u}} {\rm Pr}(\bm{v},\bm{u})N_{\bm{v},\bm{u}} \\
&\leq \sum_{\bm{v},\bm{u}} {\rm Pr}(\bm{v},\bm{u})[1+ \frac{8\Delta_{\cU}^2}{K\varepsilon^2}\ln\frac{4}{\delta}] \\
&\leq 1 + \frac{8\Delta_{\cU}^2}{K\varepsilon^2}\ln\frac{4}{\delta},
\end{align}

The total number of samples consumed is given by
\begin{align}
    \bE[N] = \sum_{k=1}^K \bE[N_k]
\leq 1+ \frac{8\Delta_{\cU}}{\varepsilon^2\delta} + 
{8\Delta_{\cU}^2}\times\frac{1}{\varepsilon^2}\ln\frac{4}{\delta}.
\end{align}
Therefore, the sample complexity is $\cO \big( \frac{2^{[\cM(\cU)]}}{\varepsilon^2\delta} + 
{2^{[2\cdot\cM(\cU)]}}\times\frac{1}{\varepsilon^2}\ln\frac{1}{\delta}\big)$ with $\Delta_{\cU} = 2^{[\cM(\cU)]}$.


\section{Proof of Proposition~\ref{pro:clifford}}\label{appx:pro:clifford}

The analysis of Protocol~\ref{protocol:clifford-fidelity-estimation-DFE} serves as the proof of Theorem~\ref{pro:clifford}. Protocol~\ref{protocol:clifford-fidelity-estimation-DFE} is a two-stage estimation procedure: 
we first estimate each $\widetilde{W}_\rho(\bm{u}_k)$ and then use them to estimate the target $\widetilde{Y}$:

\begin{align}
    O_{j\vert k} \quad\rightarrow\quad \wt{X}_{\bm{v}_k,\bm{u}_k}
\quad\rightarrow\quad \wt{Y} 
\quad\xrightarrow{\text{\color{klevinblue}~Hoeffding inequality~}}\quad Y 
\quad\xrightarrow{\text{\color{klevinblue}~Hoeffding inequality~}}\quad \frac{\tr[\cU^{\dagger}\Lambda]}{d^{2n}}.
\end{align}

\subparagraph*{Step 1: Estimating $Y$.}
Assume that each $W_{\Lambda}(\bm{v}|\bm{u})$ can be estimated \textit{ideally}. 
Now say we want to estimate $\frac{1}{d^{2n}}\tr[\cU^{\dagger}\Lambda]$ with some fixed additive error $\varepsilon/2$ 
and significance level (aka. failure probability) $\delta/2$. 
We sample $K$ times and obtain random variables $W_{\Lambda}(\bm{v}_1,\bm{u}_1),\cdots,W_{\Lambda}(\bm{v}_K,\bm{u}_K)$
and define the sample average $Y=1/K\sum_{k=1}^K W_{\Lambda}(\bm{v}_k,\bm{u}_k)$.
Since $W_{\Lambda}(\bm{v}_k,\bm{u}_k)=\in [-1,1]$ which is proved in Lemma~\ref{lemma:properties}, 
we can apply the Hoeffding inequality to obtain a tighter bound on the estimation error of $Y$.

First of all, we have
\begin{align}
  \bE[Y]  = \bE\left[\frac{1}{K}\sum_{k=1}^KW_{\Lambda}(\bm{v}_k,\bm{u}_k)\right] 
         =\frac{1}{d^{2n}}\tr[\cU^{\dagger}\Lambda].
\end{align}
where $\bE$ denotes the expected value over the random choice of $k$.
Applying the above Hoeffding inequality to $Y$ we obtain
\begin{align}
  \Pr\left\{\left\vert Y - \bE[Y]\right\vert \geq \varepsilon/2\right\}
= \Pr\left\{\left\vert Y - \frac{1}{d^{2n}}\tr[\cU^{\dagger}\Lambda]\right\vert \geq \varepsilon/2\right\}
\leq 2\exp\left( -\frac{K^2\varepsilon^2}{2\sum_{k=1}^K2^2} \right)
= 2\exp\left( -\frac{K\varepsilon^2}{8} \right).
\end{align}
Let \textbf{\color{klevinblue}$K=\lceil 8/\varepsilon^2\ln(4/\delta)\rceil$} yields
\begin{align}\label{eq:estimation-bound-3}
  \Pr\left\{\left\vert Y - \bE[Y]\right\vert \geq \varepsilon/2\right\}
\leq 2\exp\left( -\frac{K\varepsilon^2}{8} \right) = \delta/2.
\end{align}

This is the crucial difference between
Protocol~\ref{protocol:fidelity-estimation-DFE-mana} and Protocol~\ref{protocol:stabilizer-fidelity-estimation-DFE}:
In Protocol~\ref{protocol:fidelity-estimation-DFE-mana} we employ a weaker Chebyshev inequality to bound $K$, 
while in Protocol~\ref{protocol:stabilizer-fidelity-estimation-DFE}, we employ the Hoeffding inequality to bound $K$,
yielding an exponentially better dependence on significance level $\delta$.

\subparagraph*{Step 2: Estimating $W_{\Lambda}(\bm{v}_k,\bm{u}_k)$.}

However, $W_{\Lambda}(\bm{v}_k,\bm{u}_k)$ can not be ideally estimated. Thus, to complete the description of our method, we show how the ideal infinite-precision estimator $Y$
can be approximated by an estimator $\wt{Y}$ that is obtained from a finite number of copies of the unitary  $U$:
\begin{align}
    \wt{Y} 
&= \frac{1}{K}\sum_{k=1}^K\wt{W}_{\Lambda}(\bm{v}_k,\bm{u}_k) \\
&= \sum_{k=1}^K\sum_{j=1}^{N_k} \frac{1}{K}\times\frac{1}{N_k}\times O_{j\vert k} \\
&= \sum_{k=1}^K\sum_{j=1}^{N_k}\wt{O}_{jk},
\end{align}
where
\begin{align}
  \wt{O}_{j\vert k} := \frac{1}{KN_k}O_{j\vert k} \in [-1/(KN_k), 1/(KN_k)].
\end{align}
The range bound follows since $O_{j\vert k}\in\{-1,1\}$.
That is to say, $\wt{Y}$ can be viewed as a sum of random variables $\wt{O}_{j\vert k}$. 
Now we can apply another version of the Hoeffding's inequality in terms of sum, yielding:
\begin{subequations}\label{eq:estimation-bound-4}
\begin{align}
    \opn{Pr}\left(\left\vert\wt{Y} - Y \right\vert \geq \varepsilon/2\right)
&\leq 2\exp\left(- \frac{2\varepsilon^2}{4\sum_{k=1}^K\sum_{j=1}^{N_k}(u_b - l_b)^2}\right) \\
&= 2\exp\left(- \frac{\varepsilon^2}{2\sum_{k=1}^K\sum_{j=1}^{N_k}\frac{4}{K^2N_k^2}}\right) \\
&= {2\exp\left(- \frac{\varepsilon^2}{\sum_{k=1}^K\frac{8}{K^2N_k}}\right)}.
\end{align}
\end{subequations}
We need to cleverly choose $N_k$ so that the error probability satisfies
\begin{subequations}\label{eq:bound-on-Nk}
\begin{align}
&\quad    {2\exp\left(- \frac{\varepsilon^2}{\sum_{k=1}^K\frac{8}{K^2N_k}}\right)} = \delta/2 \\
\Leftarrow&\quad \sum_{k=1}^K\frac{8}{K^2N_k} = \frac{\varepsilon^2}{\ln\frac{4}{\delta}} \\
\Leftarrow&\quad  \forall k=1,\cdots, K,\quad 
                  \frac{8}{K^2N_k} = \frac{1}{K}\times\frac{\varepsilon^2}{\ln\frac{4}{\delta}} \\
\Leftarrow&\quad  \forall k=1,\cdots, K,\quad 
                  N_k = \left\lceil\frac{8}{K\varepsilon^2}\ln\frac{4}{\delta}\right\rceil.
\end{align}
\end{subequations}

\subparagraph*{Step 3: Union bound}

From the above two steps, we obtain the following two probability bounds
given in Eqs.~\eqref{eq:estimation-bound-3} and~\eqref{eq:estimation-bound-4}:
\begin{align}
  \Pr\left\{\left\vert Y - \bE[Y]\right\vert \geq \varepsilon/2\right\}
&\leq \delta/2, \\
\opn{Pr}\left(\left\vert\wt{Y} - Y \right\vert \geq \varepsilon/2\right)
&\leq \delta/2.
\end{align}
Using the union bound, we have 
\begin{align}
  \Pr\left\{\left\vert \wt{Y} - \bE[Y]\right\vert \leq \varepsilon\right\}
&= 1 - \Pr\left\{\left\vert \wt{Y} - \bE[Y]\right\vert \geq \varepsilon\right\}  \\
&\geq 1 - \left( \Pr\left\{\left\vert Y - \bE[Y]\right\vert \geq \varepsilon/2\right\}
               + \opn{Pr}\left(\vert\wt{Y} - Y \vert \geq \varepsilon/2\right)\right) \\
&\geq 1 - \delta.
\end{align}
We can then conclude that,
\begin{quote}
\textit{\color{googlered}
With probability $\geq 1-\delta$, 
the fidelity $\cF(\psi,\rho)$ lies in the range $[\wt{Y}-\varepsilon, \wt{Y}+\varepsilon]$.}
\end{quote}

\subparagraph*{Number of copies}

Notice that the DFE method samples $K=\lceil 8/\varepsilon^2\ln(4/\delta)\rceil$
Pauli observables, independent of the size of the system. 
For each sampled $k$, it requires $N_k$ number of copies, 
which is actually independent on $k$ (cf. Eq.~\eqref{eq:bound-on-Nk}),
thus the total number of samples consumed is given by
\begin{align}
    N = \sum_{k=1}^K N_k
      = \sum_{k=1}^K \left\lceil\frac{8}{K\varepsilon^2}\ln\frac{4}{\delta}\right\rceil
      \sim \left\lceil\frac{8}{\varepsilon^2}\ln\frac{4}{\delta}\right\rceil.
\end{align}
\end{document}